\documentclass [a4paper, 11 pt]{article}

\usepackage{amsmath}
\usepackage{amsfonts}
\usepackage{amsthm}
\usepackage{amscd}
\usepackage{amssymb}
\usepackage{graphicx}
\usepackage{epsfig}
\usepackage{graphics}

\newcommand{\Rr}{{\mathbb R}}
\newcommand{\Cc}{{\mathbb C}}

\newcommand{\half}{{\textstyle \frac{1}{2}}}

\def\ci{\mathrm{i}}
\newcommand{\sgn}{\,\text{sgn}\,}

\renewcommand{\exp}{\mathrm{e}}

\newcommand{\evec}{{\bf e}}

\newcommand{\eps}{{\epsilon}}
\newcommand{\svec}{{\bf s}}
\newcommand{\wbar}{{\bar w}}
\newcommand{\tbar}{{\bar t}}

\newcommand{\Ecal}{{\cal E}}
\newcommand{\Hcal}{{\cal H}}

\newcommand{\jhat}{{\hat \jmath}}

\newcommand{\Ccal}{{\cal C}}
\newcommand{\Pcal}{{\cal P}}
\newcommand{\Qcal}{{\cal Q}}
\newcommand{\Vcal}{{\cal V}}
\newcommand{\Vcalb}{{\bf \cal V}}
\newcommand{\Wcal}{{\cal W}}
\newcommand{\Rcal}{{\cal R}}

\newcommand{{\ad}}{\,\text{ad}}
\newcommand{{\adj}}{\,\text{adj}\,}

\newcommand{\nuvec}{{\boldsymbol \nu}}
\newcommand{\degree}{\,{\text {deg}}}

\newcommand{\What}{{S}}
\newcommand{\Whatb}{{\bf S}}

\newcommand{\Pcalb}{{\bf \cal P}}

\newcommand{\Ub}{{\bf U}}
\newcommand{\Ubp}{{\bf U'}}

\newcommand{\hb}{{\bf h}}
\newcommand{\hbi}{{\bf h^{-1}}}
\newcommand{\hbsi}{{\bf h_i}}

\newcommand{\Vb}{{\bf V}}
\newcommand{\Vbp}{{\bf V'}}
\newcommand{\Wb}{{\bf W}}
\newcommand{\Ical}{{\cal I}}
\newcommand{\Lcal}{{\cal L}}

\newcommand\xhat{{\hat { \mathbf{x}}}}
\newcommand\yhat{{\hat { \mathbf{y}}}}
\newcommand\zhat{{\hat { \mathbf{z}}}}
\newtheorem{thm}{Theorem}
\newtheorem{lem}{Lemma}[subsection]
\newtheorem{prop}{Proposition}[subsection]
\newtheorem{cor}{Corollary}

\theoremstyle{definition}

\theoremstyle{remark}

\def\comment#1{}

\def\withcomments{
\addtolength{\oddsidemargin}{-0.5 in}
\addtolength{\evensidemargin}{-0.5 in}
\newcounter{mycommentcounter}
\def\comment##1{\refstepcounter{mycommentcounter}%
  \ifhmode%
  \unskip%
  {\dimen1=\baselineskip \divide\dimen1 by 2 %
    \raise\dimen1\llap{\tiny -\themycommentcounter-}}\fi%
  \marginpar{\renewcommand{\baselinestretch}{0.8}%
    \footnotesize [\themycommentcounter]: \raggedright ##1}}
}

\withcomments

\begin{document}

\title{Tangent unit-vector fields: nonabelian homotopy invariants and the Dirichlet energy}
\author{A Majumdar$^{\dag}$,
JM Robbins$^\ddag$ \& M Zyskin$^\sharp$
\thanks{ majumdar@maths.ox.ac.uk, j.robbins@bristol.ac.uk, zyskin@yahoo.com}\\
$\dag$ Mathematical Institute\\ University of Oxford, 24 -- 29
St.Giles,
Oxford OX1 3LB, UK\\
and\\
$\ddag$ School of Mathematics,\\
 University of Bristol, University Walk, Bristol BS8 1TW, UK \\
 $^\sharp$ Department of Mathematics,  SETB 2.454 - 80 Fort Brown, Brownsville,
TX 78520, USA}

\thispagestyle{empty} \maketitle

\begin{abstract}
Let $O$ be a closed geodesic polygon in $S^2$.  Maps from $O$ into
$S^2$ are said to satisfy tangent boundary conditions if the edges
of $O$ are mapped into the geodesics which contain them.  Taking $O$
to be an octant of $S^2$, we compute the infimum Dirichlet energy,
$\Ecal(H)$, for continuous maps satisfying tangent boundary
conditions of arbitrary homotopy type $H$. The expression for
$\Ecal(H)$ involves a topological invariant -- the spelling length
-- associated with the (nonabelian) fundamental group of the
$n$-times punctured two-sphere, $\pi_1(S^2 - \{s_1,\ldots, s_n\},
*)$. The lower bound for $\Ecal(H)$ is obtained from combinatorial
group theory arguments, while the upper bound is obtained by
constructing explicit representatives which, on all but an
arbitrarily small subset of $O$, are alternatively locally conformal
or anticonformal. For conformal and anticonformal classes (classes
containing wholly conformal and anticonformal representatives
respectively),  the expression for $\Ecal(H)$ reduces to a previous
result involving the degrees of a set of regular values $s_1,
\ldots, s_n$ in the target $S^2$ space.  These degrees may be viewed
as invariants associated with the abelianization of $\pi_1(S^2 -
\{s_1,\ldots, s_n\}, *)$. For nonconformal classes, however,
$\Ecal(H)$ may be strictly greater than the abelian bound.  This
stems from the fact that, for nonconformal maps, the number of
preimages of certain regular values may necessarily be strictly
greater than the absolute value of their degrees.

This work is motivated by the theoretical modelling of nematic
liquid crystals in confined polyhedral geometries. The results imply
new lower and upper bounds for the Dirichlet energy (one-constant
Oseen-Frank energy) of reflection-symmetric tangent unit-vector
fields in a rectangular prism.
\end{abstract}

\section{Statement of Results}\label{sec: statement}

Let $\Ccal(S^2,S^2)$ denote the space of continuous maps of the
two-sphere into itself. As is well known,  maps in $\Ccal(S^2,S^2)$
are classified up to homotopy by their degree, and the infimum of
the Dirichlet energy, $\int_{S^2} |\phi'|^2 dA$, for maps $\phi$ of
given degree $d$ is equal to $8\pi d$ (the area element $dA$ is
normalised so that $S^2$ has area $4\pi$). Critical points of the
Dirichlet energy include conformal or anticonformal maps, and the
infimum energy may be realised by conformal maps (for $d \ge 0$) and
anticonformal maps (for $d \le 0)$ on $S^2$.

In this paper we study an elaboration of this problem, motivated, as
explained in Section~\ref{sec:LC} below, by models of nematic liquid
crystals in confined polyhedral geometries. Consider a set of $f$
geodesics (great circles) on  $S^2$.  The geodesics divide $S^2$
into a collection of closed spherical polygons (Euler's theorem
implies that, generically, there are $f^2 - f + 2$ polygons). Let
$O$ denote one such polygon.  A map $\nu: O \rightarrow S^2$ is said
to satisfy {\it tangent boundary conditions} if $\nu$ maps each edge
of $O$ into the geodesic which contains it.  Let $\Ccal_T(O,S^2)$
denote the set of continuous maps $\nu: O \rightarrow S^2$ which
satisfy tangent boundary conditions.  We may then ask, for a given
homotopy class $H$ in $\Ccal_T(O,S^2)$, what is the infimum of the
Dirichlet energy?

In this paper we will take $O$ to be the positive coordinate octant,
\begin{equation}
O = \{ r \in S^2 \subset \Rr^3\, | r
_j \ge 0\}, \label{eq:O}
\end{equation}
whose edges lie on the three coordinate geodesics, ie the unit
circles about the origin in the $xy$-, $yz$- and $zx$-planes.  The
vertices of $O$ are the coordinate unit vectors $\xhat$, $\yhat$ and
$\zhat$.

The homotopy classification of $\Ccal_T(O,S^2)$
is described in \cite{rz, mrz3,mrz4}.  Let us summarize the relevant
results.  First, given $\nu \in \Ccal_T(O,S^2)$, consider the values
of $\nu$ at the vertices of $O$. Tangent boundary conditions imply
that $\nu(\jhat) = e_j \jhat$, where $e_j = \pm 1$. The $e_j$'s,
called \emph{edge signs}, are homotopy invariants. Consider next the
values of $\nu$ on the edges of $O$. The image of the $yz$-edge (for
example) under $\nu$ is a curve on the $yz$-coordinate circle with
endpoints $e_y \yhat$ and $e_z \zhat$. The integer-valued winding
number of this curve relative to the shortest geodesic between its
endpoints is another invariant, called the \emph{kink number}, which
we denote by $k_x$. (To be consistent with our previous conventions,
we take $k_x \le 0$ if $\nu$ preserves orientation on the
$yz$-edge.)  The kink numbers $k_y$ and $k_z$ are defined similarly.
Finally, the oriented area of the image of the interior of $O$ under
$\nu$, denoted $\Omega$, is also an invariant, called the
\emph{trapped area}. For $\nu$ differentiable, the trapped area is
given by
\begin{equation}
\Omega = -\int_O \nu^* \omega, \label{eq:Omega}
\end{equation}
where $\omega$ is the area two-form on $S^2$ ($\omega$ is normalised
so that $\int_{S^2}\omega = 4\pi$).  (To be consistent with our
previous conventions, we have taken $\Omega < 0$ for $\nu$
orientation-preserving.)  For a given set of edge signs $e=\left(e_x, e_y, e_z\right)$ and kink numbers $k=\left(k_x, k_y, k_z\right)$, the allowed values of
$\Omega$ differ by integer multiples of $4\pi$. The invariants $(e,k,\Omega)$ collectively classify the
homotopy classes of $\Ccal_T(O,S^2)$, and all allowed values can be
realised.

A second classification scheme, described in \cite{mrz3,mrz4}, is
based on generalised degrees. 
Observe that the coordinate geodesics which define the domain $O$
also partition the target $S^2$ space into open coordinate octants,
here called {\it sectors}.  We label the sectors by fixing the signs
of the coordinates. Thus, we let $\sigma = (\sigma_x\, \sigma_y\,
\sigma_z)$ denote a triple of signs, and define the sector
$\Sigma_\sigma$ by
\begin{equation}
\Sigma_\sigma = \{ s \in S^2 \subset \Rr^3\, | \, \sigma_j  s_j >
0\}. \label{eq:Sigma}
\end{equation}
(Thus, $O$ is the closure of $\Sigma_{\scriptscriptstyle +++}$,
although we shall regard $O$ and $cl(\Sigma_{\scriptscriptstyle
+++})$ as distinct, with $O$ constituting the domain of $\nu$ and
$cl(\Sigma_{\scriptscriptstyle +++})$ constituting a subset of the
target space.) It can be shown that the degree of a regular value
$s_\sigma \in \Sigma_\sigma$ of $\nu$, ie, the number of preimages of
$s_\sigma$ counted with a sign according to orientation, is a
homotopy invariant called the \emph{wrapping number}, which we
denote by $w_\sigma$. For $\nu$ differentiable,
\begin{equation}
w_\sigma  = -\sum_{x \in \nu^{-1}(s_\sigma)} \sgn \det \nu'(x)
\label{eq:w_sigma}
\end{equation}
(To be consistent with our previous conventions, we have taken
$w_\sigma \le 0$ for $\nu$ orientation-preserving.) Using  Stokes'
theorem, one can express the wrapping numbers in terms of
$(e,k,\Omega)$,
\begin{equation}
\label{eq:wrap in terms of e,k,omega}
 w_\sigma = \frac{1}{4\pi}\Omega +
\frac{1}{2}\sum_j \sigma_j k_j + e_x e_y e_z\left(\frac{1}{8} -\delta_{\sigma,e}\right).
\end{equation}
(\ref{eq:wrap in terms of e,k,omega}) can be inverted to obtain
$(e,k,\Omega)$ in terms of the $w_\sigma$'s.  Thus the wrapping
numbers $w_\sigma$ are a (constrained) set  of classifying
invariants for $\Ccal_T(O,S^2)$.

We say that a homotopy class in $\Ccal_T(O,S^2)$ is {\it conformal}
if $w_\sigma \le 0$ for all $\sigma$, {\it anticonformal} if $w_\sigma \ge 0$ for
all $\sigma$, and {\it nonconformal} otherwise. In \cite{mrz3} it is
shown that every conformal homotopy class has a conformal
representative. In terms of the complex coordinate $w =
(s_x+i s_y)/(1+s_z)$ on $S^2$ (for $\svec = \left(s_x, s_y, s_z\right)\in S^2$), the conformal representatives are rational
functions whose zeros and poles satisfy constraints dictated by
tangent boundary conditions. Likewise, every anticonformal homotopy
class has an anticonformal representative which is a rational
function of $\wbar$.  Representatives for nonconformal topologies
are also discussed in \cite{mrz3}.

We say that the sectors labeled by $\sigma$ and $\sigma'$ are {\it
adjacent}, denoted $\sigma \sim \sigma'$, if $\Sigma_\sigma$ and
$\Sigma_{\sigma'}$ share a common edge, or equivalently, if $\sigma$
and $\sigma'$ have precisely two components the same.

Given $\nu \in \Ccal_T(O,S^2) \cap W^{1,2}(O,S^2)$, let
\begin{equation}\label{eq:E_2}
E(\nu) = \int_O  |\nu'|^2 \, dA
\end{equation}
denote the Dirichlet energy of $\nu$. Given a homotopy class $H
\subset \Ccal_T(O,S^2)$, let
\begin{equation}\label{eq:Ecal_2}
\Ecal(H) = \inf_{\nu \in H} E(\nu)
\end{equation}
denote the infimum Dirichlet energy over $H$. Our main result,
contained in Theorems~\ref{thm:1} and \ref{thm:2} below, is an
explicit formula for $\Ecal(H)$.  The formula consists of two
contributions. The first, $\sum_\sigma |w_\sigma| \pi$,
which on its own constitutes a lower bound for $\Ecal(H)$,
follows from considerations of the algebraic degree. This is analogous to what one
has for maps in $\Ccal(S^2,S^2)$.

The additional contribution involves a new homotopy invariant,
$\Delta(H)$, which we now define.  First, we have that
\begin{equation}\label{eq: Delta conformal}
    \Delta(H) = 0, \ \  H \ \text{conformal or anticonformal}.
\end{equation}
For $H$ nonconformal,  let $\sigma_{+}$ label a sector with the
largest positive wrapping number (in cases where $\sigma_+$ is not
unique, the definition (\ref{eq:Delta(H) neq 0}) below does not
depend on the choice of $\sigma_+$.) Similarly,
 let $\sigma_{-}$ denote the
sector with the smallest negative wrapping number (ie, negative
wrapping number of largest magnitude). Let
\begin{equation}\label{eq: chi}
\chi =
\begin{cases}
1, & k_x k_y k_z < 0,\\
0, & \text{otherwise}.
\end{cases}
\end{equation}
Then we define
  \begin{equation}
    \label{eq:Delta(H) neq 0}
    \Delta(H) =
    2\max\left(0,\ w_{\sigma_+}  -
    \sum_{\sigma \sim \sigma_+}
    \Phi(w_\sigma) - \chi,\ |w_{\sigma_{-}}| -
    \sum_{\sigma \sim \sigma_{-}}
    \Phi(-w_\sigma) - \chi \right), \ \ H \
    \text{nonconformal},
  \end{equation}
  where
 \begin{equation}\label{eq:Phi}
\Phi(x) = \half(x + |x|).
   \end{equation}

We now state our main results.
\begin{thm} \label{thm:1}
Let $H$ be a homotopy class in $\Ccal_T(O,S^2)$.  Then
\begin{equation}
  \label{eq: lower bound}
 \Ecal(H)  \ge \left(\sum_\sigma |w_\sigma|  + \Delta(H)  \right)\pi.
\end{equation}
\end{thm}
\begin{thm} \label{thm:2}
Let $H$ be a homotopy class in $\Ccal_T(O,S^2)$.  Then
\begin{equation}
  \label{eq: upper bound}
 \Ecal(H)  \le \left(\sum_\sigma |w_\sigma| + \Delta(H)  \right)\pi.
\end{equation}
\end{thm}
\begin{cor}\label{cor: equation}
Let $H$ be a homotopy class in $\Ccal_T(O,S^2)$.  Then
\begin{equation}
  \label{eq: equality}
 \Ecal(H)  = \sum_\sigma |w_\sigma|\pi  + \Delta(H)  \pi.
\end{equation}
\end{cor}

Theorem~\ref{thm:1} is proved in Section~\ref{sec:lb}.  There the
quantity $\sum_\sigma |w_\sigma|  + \Delta(H)$ is related to an
invariant of the (nonabelian) fundamental group $\pi_1(S^2 - S,
*)$, where $S$ is a set of representative points in $S^2$ from
four appropriately chosen sectors (the choice is determined by $H$).
Theorem \ref{thm:2} is proved in Section~\ref{sec: upper bound} by
constructing a sequence of maps $\nu_\epsilon \in H$ whose energy
approaches the upper bound (\ref{eq: upper bound}) as $\epsilon$
approaches $0$. The maps $\nu_\epsilon$ are alternatively locally
conformal or locally anticonformal except on a set whose area
vanishes with $\epsilon$.

Theorems~\ref{thm:1} and \ref{thm:2} follow from new methods compared to our previous work in \cite{mrz2, mrz3}. The derivation of the bounds (\ref{eq: lower bound}) and (\ref{eq: upper bound}) involve combinatorial-group-theoretic arguments and non-trivial explicit constructions. For conformal and anticonformal topologies, Corollary~\ref{cor: equation}
coincides with  results given in \cite{mrz2}.  For nonconformal
topologies, Corollary~\ref{cor: equation} is a sharp improvement of estimates
obtained in \cite{mrz3}, which are equivalent to $\sum_\sigma
|w_\sigma|\pi \le \Ecal(H) \le 9 \sum_\sigma |w_\sigma|\pi$.

\subsection{Nematic liquid crystal configurations in a rectangular
prism}\label{sec:LC}

In the Oseen-Frank theory \cite{dg, virga, stewart}, the local
orientation of a nematic liquid crystal in a domain $P \subset
\Rr^3$ is described by a director field $n: P\rightarrow RP^2$.
Equilibrium configurations are local minimizers of an energy
functional with energy density quadratic in $\nabla n$ and
parameterised by three material-dependent constants.  In the
so-called one-constant approximation, the Oseen-Frank energy density
reduces to the Dirichlet energy density $(\nabla n)^2$.  For
$\Omega$ simply connected and $n$ continuous in the interior of
$\Omega$, $n$ may be assigned an orientation in a continuous way,
and may be regarded as a continuous unit-vector field  (ie,
$S^2$-valued)  on $P$.  We assume this to be the case in what
follows.

The equilibrium configurations depend crucially on the boundary
conditions.  For certain materials, {\it tangent boundary
conditions} are appropriate, according to which $n$ is required to
be tangent to the boundary $\partial P$, which is assumed to be piecewise
smooth.  Let $\Ccal_T(P, S^2)$ denote the space of continuous
unit-vector fields on $P$ which satisfy tangent boundary conditions.

In a series of papers \cite{mrz1} -- \cite{mrz6} we have studied the
case where $P$ is a polyhedral domain.  One motivation are certain
prototype designs for bistable liquid crystal displays, in which
polygonal and polyhedral geometries  support multiple equilibrium
configurations with different optical properties \cite{kg}.  The
homotopy classification of $\Ccal_T(P,S^2)$ is described in
\cite{rz, mrz4}, and a lower bound for the infimum Dirichlet energy
in terms of generalised minimal connections is obtained in
\cite{mrz4}. For a review, see \cite{mrz6}.

A number of results concern the case where $P$ is a right
rectangular prism,
\begin{equation}\label{eq: prism}
P = \{ r \in \Rr^3 \, | \, 0 \le r_j \le L_j\}.
\end{equation}
For definiteness, we label the sides so that $L_z \le L_y \le L_x$.
Let $L = (L_z^2 + L_y^2 + L_z^2)^{1/2}$ denote the length of the
prism diagonal. We have considered in particular {\it
reflection-symmetric} homotopy classes in $\Ccal_T(P,S^2)$.  We say
that a configuration $n \in \Ccal_T(P,S^2)$ is reflection-symmetric
if it is invariant under reflection through the midplanes of $P$, ie
\begin{equation}\label{eq: reflection-symmetric}
n(x,y,z) = n(L_x-x,y,z) = n(x,L_y-y, z) = n(x,y,L_z-z).
\end{equation}
$n$ is therefore determined by its restriction to a fundamental
domain with respect to reflections, eg
\begin{equation}\label{eq: prism_8}
R = \{ r \in \Rr^3 \, | \, 0 \le r_j \le \half L_j\}.
\end{equation}
A homotopy class $h \subset \Ccal_T(P,S^2)$ is reflection-symmetric
if (and only if) it contains a reflection-symmetric representative.
For reflection-symmetric homotopy classes, the infimum of the
Dirichlet energy is given by
\begin{equation}\label{eq:Ecal_3}
\Ecal_3(h) = \inf_{n \in h} 8 \int_R (\nabla n)^2 \, dV.
\end{equation}

Given $n \in \Ccal_T(P,S^2)$ and $0 < a < L_z$, we define a map
$\nu_{n,a}: O\rightarrow S^2$ by restricting $n$ to the surface $|r|
= a$ in $P$, ie $\nu_{n,a}(s) = n(as)$. It is readily established
that i) $\nu_{n,a} \in \Ccal_T(O,S^2)$, ii) the homotopy class of
$\nu_{n,a}$ is independent of $a$, and iii) reflection-symmetric
homotopy classes in $\Ccal_T(P,S^2)$ are in 1-1 correspondence with
the homotopy classes of $\Ccal_T(O,S^2)$.  Let $H$ denote the
homotopy class of $\Ccal_T(O,S^2)$ corresponding to $h$.

 Theorems \ref{thm:1} and \ref{thm:2}
imply lower and upper bounds on $\Ecal_3(h)$ for reflection-symmetric homotopy classes $h$.  Given $\nu \in H
\subset \Ccal_T(O,S^2)$, we construct a reflection-symmetric $n\in
\Ccal_T(P,S^2)$ via $n(r) = \nu(r/|r|)$ for $r \in R$.  Then $n \in
h$, and $r^2|\nabla n|^2(r) = |\nu'|^2$. The integral over $R$ in
(\ref{eq:Ecal_3}) is bounded above by an integral over $r < L/2$,
leading to the inequality $\Ecal_3(h) \le 4 L \Ecal(H)$. Conversely,
for any $n \in h$, we have that $r^2|\nabla n|^2 \ge |\nu'|^2$.  As the integral over $R$ in (\ref{eq:Ecal_3}) is
bounded below by an integral over $r < L_z/2$, it follows that
$\Ecal_3(h) \ge 4 L_z \Ecal(H)$.  We summarize these results in the
following - 
\begin{cor}\label{cor:2}
Let $P$ be the right rectangular prism (\ref{eq: prism}) with edge-lengths $L_z
\le L_y \le L_x$ and diagonal length $L$. Let $h$ denote a
reflection-symmetric homotopy class in $\Ccal_T(P,S^2)$, and $H$ the
corresponding homotopy class in $\Ccal_T(O,S^2)$. Then
\begin{equation}\label{eq:cor1}
    4L_z \Ecal(H) \le \Ecal_3(h) \le 4L \Ecal(H).
\end{equation}
\end{cor}
For conformal and anticonformal homotopy classes,
Corollary~\ref{cor:2} coincides with the results of \cite{mrz2}, and
for nonconformal homotopy classes  constitutes a sharp improvement of a
result from \cite{mrz3}.

Brezis, Coron and Lieb obtained the infimum Dirichlet energy for
$S^2$-valued maps on $\Rr^3$ with prescribed degrees on a set of
excluded points, or defects (they also considered more general
domains in $\Rr^3$ with holes) \cite{bcl}.  Their result is
expressed in terms of the length of a {\it minimal connection}, ie a
pairing between defects  of opposite sign. The estimates of
\cite{mrz1} -- \cite{mrz4} may be regarded as extensions of this
classical result to the case of polyhedral domains with tangent
boundary conditions, in which there are necessarily singularities at
vertices. Our previous estimates may be expressed as a sum over
minimal connections, one for each sector of the target $S^2$ space,
between the vertices of the polyhedral domain.  The new lower bound in Theorem~\ref{thm:1} contains additional topological information not captured by the minimal connection theory in \cite{bcl}. In particular, the new homotopy invariant, $\Delta(H)$ in (\ref{eq:Delta(H) neq 0}), elucidates the fact that for certain nonconformal homotopy classes, the absolute number of pre-images of a regular value may necessarily be greater than the absolute value of $|w_\sigma|$. In such cases, the infimum energy is necessarily greater than the abelian bound, $\pi\sum_\sigma |w_\sigma|$, predicted by minimal connection theory. It would be
interesting to generalise Corollary~\ref{cor:2} to
non-reflection-symmetric configurations on $P$ as well as to more
general polyhedral domains. Results in this direction may involve a
nonabelian extension of the notion of minimal connection.

\section{Lower bound for $\Ecal(H)$}\label{sec:lb}

Given $\nu \in \Ccal_T(O,S^2) \cap W^{1,2}(O,S^2)$, we can obtain a
lower bound for the Dirichlet energy $E(\nu)$
in terms of the number of preimages of a set of regular values of
$\nu$, one from each sector of $S^2$ (Lemma~\ref{lem:1},
Section~\ref{sec: lb and abs degree}).  This leads to the following
problem, which is addressed in Section~\ref{sec: abs degree S2}:
given a smooth unit-vector field $\mu$ on the two-disk $D^2$ for
which the homotopy class of the  boundary map $\partial \mu$ is
prescribed, find a lower bound for the number of preimages of a
finite set $S$ of regular values of $\mu$.  The bound is expressed
as the infimum of a certain function -- the {\it spelling length} --
over a product of conjugacy classes in the fundamental group
$\pi_1(S^2 - S, *)$. The bound is obtained by analysing a simpler
problem in Section \ref{sec: abs degree R2}, in which the target
space is taken to be $\Rr^2$ rather than $S^2$. The estimates of the
spelling lengths relevant to our problem are given in
Section~\ref{sec: Lb for Ecal}, yielding a proof of
Theorem~\ref{thm:1}.

Let us introduce some notation. Let X and Y be two-dimensional manifolds, 
possibly with boundary.  Let $int(X)$ denote the interior of $X$ 
(similarly $int(Y)$).  Let $f : X \to Y$ be piecewise continuously 
differentiable (in Section~\ref{sec: upper bound} it will be convenient to allow for maps with 
piecewise continuous derivatives).  We say that $y \in Y$ is a regular 
value of $f$ if and only if $y \in int(Y)$, $f^{-1}(y) \subset int(X)$, 
and $f'$ is continuous and of full rank at each point of $f^{-1}(y)$.  Let 
$\Rcal_f$ denote the set of regular values of $f$. We recall Sard's theorem
\cite{spivak}, according to which $\Rcal_f$ has full Lebesgue
measure. For $y \in \Rcal_f$, let
\begin{align}
d_f(y) &= \sum_{x \in f^{-1}(y)} \sgn \det f'(x),\label{eq: d_f}\\
D_f(y) &= \sum_{x \in f^{-1}(y)} 1.
\end{align}
$d_f(y)$ is the algebraic degree, or simply the degree, of $y$, ie the number
of preimages of $y$ counted with orientation.   $D_f(y)$, on the other hand, is the number of preimages of y. For convenience, we will refer to $D_f(y)$ as the {\it
absolute degree} of $y$ although it should not be confused with the Hopf absolute degree \cite{hopf} which is used elsewhere in the literature. We remark that $d_f(y)$ is invariant under differentiable deformations of $f$ (provided $y$ remains a regular value), whereas $D_f(y)$ is not.
Clearly
\begin{equation}\label{eq: d and D}
    |d_f(y)| \le D_f(y).
\end{equation}
Wrapping numbers are examples of algebraic degrees.  Indeed,
for $\nu \in \Ccal_T(O,S^2)$ differentiable and $s_\sigma \in
\Sigma_\sigma$ a regular value of $\nu$, we have that
\begin{equation}\label{eq: d and w}
    d_\nu(s_\sigma) = -w_\sigma.
\end{equation}

\subsection{Lower bound and absolute degree}\label{sec: lb and abs degree}

\begin{lem} \label{lem:1}
Let $\nu \in \Ccal_T(O,S^2)$ be differentiable.  For each $\sigma$, let
$s_\sigma \in \Sigma_\sigma \cap \Rcal_\nu$.
Then
 \begin{equation}
 \label{eq:r22}
 E(\nu) 
 \ge
 \inf_{ \{s_\sigma\}} \sum_{\sigma}
 D_{\nu}(s_\sigma) \pi.
 \end{equation}
 \end{lem}

\begin{proof}
From the inequality $a^2 + b^2 + c^2 + d^2 \ge 2|ad -
bc|$, it follows that $\left|\nu'\right|^2 \ge 2 |\det \nu'|$. Then
\begin{equation}\label{eq: lem 1 eq 1}
E(\nu) = \int_{p\in O} \left|\nu'\right|^2\, dA_p \ge 2 \int_{p\in O} |\det
\nu'| \, dA_p =  2\int_{p \in O} |\det \nu'| \left( \sum_\sigma
\int_{s \in
\Sigma_\sigma}  
\delta_{S^2}(s,\nu(p))\, dA_s \right)
\, dA_p,
\end{equation}
where $\delta_{S^2}(s,t)$ is the Dirac delta function on $S^2$
normalised to have unit integral.  We may interchange the $s$- and
$p$-integrals (this can be justified by introducing smoothed delta
functions, appealing to Fubini's theorem, and taking the limit as
the smoothing parameter goes to zero).  For $s \in
\Rcal_\nu$, we have that
\begin{equation}\label{eq: lem 1 eq 2}
\int_{p \in O}  |\det \nu'(p)| \,\delta_{S^2}(s,\nu(p))  \, dA_p  =
D_\nu(s).
\end{equation}
By Sard's theorem,
the set of regular values is of
full measure.  It follows from (\ref{eq: lem 1 eq 1}) and (\ref{eq:
lem 1 eq 2}) that
\begin{equation}\label{eq: lem 1 eq 3}
E(\nu) \ge 2 \sum_\sigma \int_{s \in \Rcal_\nu \cap \Sigma_\sigma}
D_\nu(s)  \, dA_s \ge  \inf_{\{s_\sigma\}} \sum_{\sigma}
D_\nu(s_\sigma)\pi,
\end{equation}
as the sectors $\Sigma_\sigma$ each have area $\pi/2$.
 \end{proof}

\subsection{Absolute degree of $\Rr^2$-valued maps on $D^2$}\label{sec: abs degree R2}

Let $D^2 \subset \Rr^2$ denote the unit disk  with boundary
$\partial D^2 = S^1$. Let $R=\left\{ y_1, \ldots, y_n\right\}$
denote a set of $n$ distinct points in $\Rr^2$.   Let
$\pi_1\left(\Rr^2 - R, q\right)$ denote the fundamental group of the
$n$-times punctured plane, $\Rr^2 - R$,  based at $q \in \Rr^2$, where
$q \notin R$. $\pi_1(\Rr^2-R,q)$ may be identified with the free
group on $n$ generators, $F(c_1,\ldots,c_n)$ (see, eg,
\cite{magnus}). We shall take the generator $c_j$ to be the homotopy
class of a loop $\gamma_j$ based at $q$ which encircles $y_j$ once
anticlockwise but encloses no  other points of $R$. Equivalently,
$\gamma_j$ is  freely homotopic in $\Rr^2 - R$ to an
$\epsilon$-circle about $y_j$ oriented anticlockwise (with
$\epsilon$ small enough so that no other points of $R$ are contained
inside). It is straightforward to show that this condition
determines $c_j$ up to conjugacy. That is, if $\gamma$ and  $\gamma'$
 are two loops in
$\Rr^2 - R$ based at $q$ which are freely homotopic to an
anticlockwise-oriented $\epsilon$-circle about $y_j$, then
\begin{equation}
  \label{eq:conjugate_generators}
  [\gamma'] = h [\gamma] h^{-1}
\end{equation}
for some $h \in \pi_1(\Rr^2-R,q)$.

Given  $g \in F(c_1,\ldots,c_n)$ expressed as a product of the
generators, the difference between the number of $c_i$ and
$c_i^{-1}$ factors is well defined, and is called the {\it degree of
$c_i$ in $g$}, and denoted by $\degree_g(c_i)$. Given  $g \in
F(c_1,\ldots,c_n)$, we define a {\it spelling} to be a factorisation
of $g$ into a product of conjugated generators and inverse
generators, eg
\begin{equation}
  \label{eq:spelling}
  g = h_1 c_{i_1}^{\epsilon_1} h_1^{-1} \cdots h_r c_{i_r}^{\epsilon_r} h_r^{-1},
\end{equation}
where $h_j \in  F(c_1,\ldots,c_n)$ and $\epsilon_j = \pm 1$.  It is
clear that
\begin{equation}\label{eq: spelling and degree}
    \sum_{s\, |\, i_s = j} \epsilon_s = \degree_g(c_j).
\end{equation}

The number of factors in a spelling of $g$, (i.e.~$r$ in
(\ref{eq:spelling})), is not uniquely determined. We define the {\it
spelling length} of $g$, denoted $\Lambda(g)$, to be the smallest
possible number of factors
amongst all spellings of $g$.
From (\ref{eq: spelling and degree}) it follows that the spelling
length is determined modulo 2 by the sum of the degrees of the
generators,
\begin{equation}\label{eq: spelling mod 2}
    \Lambda(g) = \sum_{j=1}^n \degree_{g}(c_j) \mod 2,
\end{equation}
and is bounded from below by the sum of their absolute values,
\begin{equation}\label{eq: spelling length and degree}
    \Lambda(g) \ge \sum_{i=1}^n |\degree_{g}(c_i)|.
\end{equation}
We refer to (\ref{eq: spelling length and degree}) as the abelian
bound on the spelling length.

Let $\phi: D^2 \rightarrow \Rr^2$ be differentiable, and let
$\partial \phi: S^1 \rightarrow \Rr^2-\Rcal_\phi$ denote the
boundary map of $\phi$. Choose the points $y_j$ above to be regular
values of $\phi$, ie $y_j \in \Rcal_\phi$, and take $q$ to lie in
the image of $\partial \phi$. We may regard $\partial \phi$ as a
loop in $\Rr^2 - R$ based at $q$. We denote its homotopy class by
$[\partial \phi] \in \pi_1(\Rr^2 - R, q)$.  As the following shows,
the spelling length of $[\partial \phi]$ gives a lower bound on the
cardinality of $\phi^{-1}(R)$.

\begin{prop}\label{prop:disk}
Given $\phi: D^2\rightarrow \Rr^2$ smooth, $R = \{y_1,\ldots,y_n\}
\subset \Rcal_\phi$, and $\pi_1(\Rr^2-R,q) \backsimeq
F(c_1,\ldots,c_n)$, with generators $c_j$ as above.
 Then
\begin{equation}
  \label{eq:disk}
  \sum_{j=1}^n D_\phi(y_j) \ge \Lambda([\partial \phi]).
\end{equation}
\end{prop}

\begin{proof}
Let $N = \sum_{j=1}^n D_\phi(y_j)$, so that $N$ is
the number of points in $\phi^{-1}(R)$.
Below we argue that
$[\partial \phi]$ can be expressed as a product of $N$ factors,
\begin{equation}
  \label{eq:spelling_of_dphi}
  [\partial \phi] = [\gamma_1] \cdots [\gamma_N],
\end{equation}
in which each factor
is conjugate to a generator or an inverse generator of
$F(c_1,\ldots,c_n)$.
Then (\ref{eq:spelling_of_dphi}) constitutes a spelling of $[\partial
\phi]$ of length $N$, and (\ref{eq:disk}) follows from the
definition of the spelling length.

To establish the spelling (\ref{eq:spelling_of_dphi}), let $\phi^{-1}(R) = \{x_1,\ldots,
x_N\}$.  We note that $x_j$ is in the interior of $D^2$.
%
Take $p \in \partial D^2$ such that $\phi(p) = q$.
We regard $S^1 = \partial D^2$  as a loop based
at $p$.  As indicated in Figure \ref{fig:1}, while
keeping $p$ fixed, we can continuously deform $\partial D^2$
into a concatenation of $N$ loops based at $p$, each of which
encloses one of the $x_a$'s once (in the anticlockwise sense) and encloses none of the other
$x_a$'s.
The image of this deformation under $\phi$
yields a homotopy from $\partial \phi$ to a concatenation of $N$
 loops $\gamma_a$ in $\Rr^2 - R$ based at $q$,
each of which is freely homotopic
 in $\Rr^2-R$ to an oriented $\epsilon$-circle about  $y_{j_a} =
 \phi(x_a)$.  From (\ref{eq:conjugate_generators}),
$[\gamma_a]$ is conjugate in  $\pi_1(\Rr^2 - R,q)$ to a generator
$c_j$ or an inverse generator $c_j^{-1}$, depending on the orientation of $\gamma_a$.

\end{proof}

  \begin{figure}
  \begin{center}
  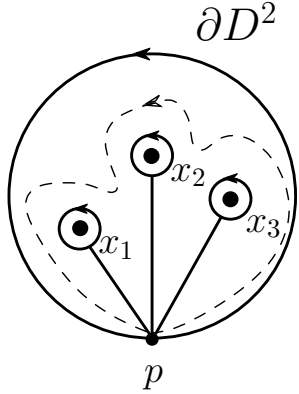
  \caption{The boundary of the two-disk, regarded as a  loop based
    at $p$, can be deformed into a concatenation of loops based
    at $p$ encircling each of the preimages $x_1, \ldots x_N$.}
  \label{fig:1}
  \end{center}
  \end{figure}

For $g = [\partial\phi]$, $\degree_{g}(c_j)$
is equal to  $d_\phi(y_j)$. Combining (\ref{eq: spelling length and degree}) and (\ref{eq:disk}), we have the following sequence of inequalities
$$ \sum_{j=1}^n D_\phi(y_j) \ge \Lambda([\partial \phi]) \ge \sum_{i=1}^n |d_\phi(y_j)|.$$
An example
where the inequality is strict is $g = c_1 c_2 c_1^{-1} c_2^{-1}$;
in this case it is easy to show that $\Lambda(g) = 2$ while
$\degree_{g}(c_1) = \degree_{g}(c_2) = 0$. 

\subsection{Absolute degree of $S^2$-valued maps on $D^2$}\label{sec: abs degree S2}

Let $\mu: D^2 \rightarrow S^2$ be a differentiable $S^2$-valued map on
$D^2$ with boundary map $\partial \mu: S^1 \rightarrow S^2$.
%
Let $S  =\left\{s_0, s_1,\ldots, s_{n}\right\} \subset \Rcal_\mu$
denote a set of $n+1$ regular values of $\mu$.  By analogy with
Proposition~\ref{prop:disk}, we seek a lower bound on the number of
points in $\mu^{-1}(S)$. In contrast to Proposition~\ref{prop:disk},
the bound we obtain will depend not only on the homotopy class of
$\partial \mu$, but also on the absolute and algebraic degrees of
one of the $s_j$'s, which we fix to be $s_0$. The bound is obtained
by excising a neighbourhood of $\mu^{-1}(s_0)$ from $D^2$ and
defining an $\Rr^2$-valued map on the remainder to which
Proposition~\ref{prop:disk} can be applied.

Let $\Pi$  denote the  projection  from $S^2 -
\{s_0\}$ to $\Rr^2$, with $s_0$ corresponding to the point at
infinity. For $1 \le j \le n$, let $r_j = \Pi(s_j)$,   and let $R =
\{r_1,\ldots, r_n\}$.  Also, take $u \in S^2$ in the image of $\partial \mu$, and let
$q = \Pi(u)$. Then, since $\Pi: S^2 - \{s_0\} \rightarrow \Rr^2$ is a diffeomorphism,
\begin{equation}\label{eq: iso_fundamental groups}
    \pi_1(S^2 - S, u) \cong \pi_1(\Rr^2 - R, q) \cong F(c_1,\ldots,
    c_n),
\end{equation}
where, as in Section~\ref{sec: abs degree R2}, the generator $c_j$
is the homotopy class of an anticlockwise loop $\gamma_j$ in $\Rr^2
- R$ based at $q$ which encloses $r_j$ once and encloses none of the
other $r_k$'s. Equivalently, we may regard $c_j$ as the homotopy
class of a  loop $\delta_j$ in $S^2 - S$ based at $u$ which
separates $s_j$ from the other $s_k$'s and is positively oriented
with respect to $s_j$.   Indeed, we may take $\delta_j =
\Pi^{-1}(\gamma_j)$.
In what follows,
we regard $[\partial \mu]$ as an
element of $F(c_1,\ldots,c_n)$.

Let $\delta_0$ be a loop in $S^2 - S$ based at $u$ which separates $s_0$
from the other $s_j$'s and is positively oriented with
respect to $s_0$.
Then $\gamma_0 = \Pi(\delta_0)$ is a loop in $\Rr^2$ based at $q$
which encloses each of the $r_j$'s once in the clockwise sense.
Let $c_0 = [\gamma_0]$ 
    denote its homotopy class.  Then $c_0$ may
be expressed as a product of the $c_j$'s in which the sum of the
exponents of each of the $c_j$'s is equal to $-1$.

We shall use the following notation.  Given subsets $V$ and $W$ of a group
$G$, we define their {\it set product} $VW \subset G$ by
\begin{equation}\label{eq: product of normal sets}
    VW = \{ v w \, | \, v \in V, w \in W\}.
\end{equation}
We denote the $n$-fold product of $V$ with itself by $V^n$. Clearly,
 if $V$ and $W$ are invariant under conjugation,
ie $h V h^{-1} = V$ for all $h \in G$ and similarly for $W$, then
$VW$ is invariant under conjugation, in which case the set product is commutative, ie
\begin{equation}\label{eq:commutative set product}
    VW = WV.
\end{equation}
Given $g \in G$, let $\langle g \rangle$ denote its conjugacy class, ie
\begin{equation}\label{eq: conjugacy class}
    \langle g \rangle = \{ g' \in G \, | \, g' = h g h^{-1} \ \text{for some} \ h \in G\}.
\end{equation}
Clearly $\langle g \rangle$ is invariant under conjugation, so the set product of
conjugacy classes is commutative.

The following gives a lower bound for the number of points in
$\mu^{-1}(S)$, given the absolute and algebraic degrees of $s_0$:
\begin{prop}\label{prop:sphere}
Let $P = \half (D_\mu(s_0) + d_\mu(s_0))$ and $N = \half (D_\mu(s_0)
- d_\mu(s_0))$ denote the number of points in $\mu^{-1}(s_0)$ with
positive and negative orientation respectively. Let $\langle c_0
\rangle $ denote the conjugacy class of $c_0$ in
$F(c_1,\ldots,c_n)$, and let $\Vcal_{P,N} \subset F(c_1,\ldots,c_n)$
be the set product
 given by
\begin{equation}\label{eq: VPN}
    \Vcal_{P,N} = \{[\partial \mu]\} \langle c_0^{-1} \rangle^{P} \langle c_0 \rangle^{N}.
\end{equation}
 Then
\begin{equation}
  \label{eq:sphere}
\sum_{j=1}^{n}D_{\mu}(s_j) \ge \min_{g \in \Vcal_{P,N}} \Lambda(g).
\end{equation}

\end{prop}
\noindent Thus, Proposition~\ref{prop:sphere} implies that
\begin{equation}\label{eq: sphere2}
\sum_{j=0}^{n}D_{\mu}(s_j) \ge D_\mu(s_0) + \min_{g \in \Vcal_{P,N}}
\Lambda(g).
\end{equation}
\begin{proof}

As we show below, by excising a suitable neighbourhood of
$\mu^{-1}(s_0)$, we can construct a differentiable map $\mu_{P+N}: D^2 \rightarrow S^2$
such that
\begin{align}\label{eq: mu_P+N props}
i)&\ \  D_{\mu_{P+N}}(s_j) = D_\mu(s_j), \ \ 1 \le j \le n, \nonumber\\
ii)&\ \ 
[\partial \mu_{P+N}] \in \Vcal_{P,N}, \nonumber\\
iii)&\ \  \mu_{P+N}^{-1}(s_0) \ \text{is empty}.
\end{align}
In view of iii), the $\Rr^2$-valued map $\phi = \Pi\circ \mu_{P+N}$
is differentiable on $D^2$, with i) $D_\phi(r_j) = D_{\mu}(s_j)$ for
$1 \le j \le n$ and ii) $ [\partial \phi] \in \Vcal_{P,N}$. Then the
claim (\ref{eq:sphere}) follows directly from
Proposition~\ref{prop:disk}, since
\begin{equation}\label{eq: [partial phi]}
\sum_{j=1}^n D_\mu(s_j) =  \sum_{j = 1}^n D_{\phi}(r_j) \ge
\Lambda([\partial \phi])\ge \min_{g \in \Vcal_{P,N}} \Lambda(g).
\end{equation}

The construction of $\mu_{P+N}$ proceeds inductively.  For $0 \le i \le
P+N$, we construct a differentiable map $\mu_i: D^2 \rightarrow
S^2$ such that
\begin{align}\label{eq: mu_i props}
i)&\ \  D_{\mu_{i}}(s_j) = D_\mu(s_j), \ \ 1 \le j \le n, \nonumber\\
ii)&\ \ [\partial
\mu_{i}] \in \Vcal_{p_i,n_i}, \nonumber\\
iii)&\ \ \text{$\mu_i$ has $P-p_i$ (resp.~$N- n_i$)
pre-images of $s_0$ with positive (resp.~negative) orientation,}\nonumber\\
&\ \ \text{with $0 \le p_i \le P$,
$0 \le n_i \le N$ and $p_i + n_i = i$.}
\end{align}
For $i = P+N$, it is evident that $\mu_{P+N}$ satisfies (\ref{eq: mu_P+N props}).

Here is the construction. For $i = 0$,
we take $\mu_0 = \mu$, with $p_0 = 0$ and $n_0 = 0$.  Then
$\mu_0$ satisfies (\ref{eq: mu_i props}) trivially.
Next, given $\mu_i$ satisfying (\ref{eq: mu_i props}) with $0 \le i < P+N$,
we construct $\mu_{i+1}$ as follows.
Take $x \in \mu_i^{-1}(s_0)$ and let $\sigma = \sgn \det \mu'_i(x)$.
Take $\epsilon
> 0$ and take $\tilde{u}^\epsilon$ to be the point on
$S^1 = \partial D^2$ at a distance $\epsilon$ anticlockwise from
$\tilde{u}$, where $\partial \mu \left(\tilde{u}\right) = u$.
Let $L^\epsilon$ be a non-self-intersecting differentiable curve from $x$ to $\tilde{u}^\epsilon$ which,
apart from its endpoints, lies in the interior of $D^2$ and contains no points in $\mu_i^{-1}(S)$.
Choose $\epsilon$ sufficiently small so
that $U^\epsilon$, the open $\epsilon$-neighbourhood of $L^\epsilon$,
contains no points in $\mu_i^{-1}(S)$ other than
$x$. See Figure~\ref{fig:2}. The boundary of $U^\epsilon$, oriented
clockwise, may be regarded as a loop based at $\tilde{u}$ which encloses a
single point in $\mu^{-1}(s_0)$
and encloses no other points in $\mu^{-1}(S)$.   It
follows that $[\mu_i(\partial U^\epsilon)]$, regarded as an element
of $F(c_1, \ldots, c_n)$, is conjugate to $c_0^{-\sigma}$.

 \begin{figure}
 \begin{center}
 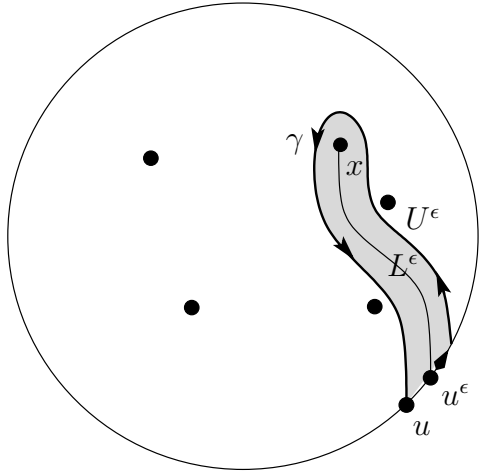
 \caption{$x$ is one of the points in $\mu^{-1}(s_0)$.  $U^\epsilon$
 is the $\epsilon$-neighbourhood of a curve $L^\epsilon$ from
 $u^\epsilon$ to $x$, with $\epsilon$ small enough so that $x$ is the
 only point in $\mu^{-1}(S)$ which lies in $U^\epsilon$. }
 \label{fig:2}
 \end{center}
 \end{figure}

The domain $D^2 - U^\epsilon$ is homeomorphic to $D^2$; let $f: D^2
\rightarrow D^2 - U^\epsilon$ be a homeomorphism.  We may take $f$
to be a diffeomorphism on the interior of $D^2$. Take $\mu_{i+i} =
\mu_i \circ f$. By construction and the induction hypothesis,
$D_{\mu_{i+1}}(s_j) = D_{\mu_i}(s_j) = D_\mu(s_j)$ for $1 \le j \le
n$.  Let $P-p_{i+1}$ and $N-n_{i+1}$ denote the number of points in
$\mu^{-1}_{i+1}(s_0)$ with positive and negative orientation
respectively.  By construction, if $\sigma = 1$, we have that
$p_{i+1} = p_i + 1$, $n_{i+1} = n_i$, while if $\sigma = -1$, we
have that $p_{i+1} = p_i$, $n_{i+1} = n_i  + 1$.  In either case, by
induction, $p_{i+1} + n_{i+1} = i+1$. Also by construction,
$\partial \mu_{i+1}$ is homotopic to the concatenation of $\partial
\mu_i$ and $\mu_i(\partial U^\epsilon)$. Since $\left[\partial \mu_i\right] \in
\Vcal_{p_i,n_i}$ (by induction) and $[\mu_i(\partial U^\epsilon)]
\in \langle c_0 \rangle^{-\sigma}$, it follows that $[\partial
\mu_{i+1}] \in \Vcal_{p_i,n_i} \langle c_0 \rangle^{-\sigma} =
\Vcal_{p_{i+1},n_{i+1}}$. So $\mu_{i+1}$ satisfies (\ref{eq: mu_i
props}).

\end{proof}

There exist efficient algorithms for computing the spelling length \cite{private communication}.
However, we are not aware of general results for obtaining the
minimum spelling length over a product of conjugacy classes. In the
cases that arise in Section~\ref{sec: Lb for Ecal}, we are
nevertheless able to compute an effective lower bound for the
spelling length on $\Vcal_{P,N}$ (cf Propositions~\ref{prop:
spelling length} and \ref{prop: spelling length 2}).

\subsection{Proof of Theorem~\ref{thm:1}} \label{sec: Lb for Ecal}

\begin{proof}

Let $H$ be a homotopy class in $\Ccal_T(O,S^2)$, with invariants $(e,k,\Omega)$ and
$\{w_\sigma\}$.  Given $\nu \in H$,
we show that
\begin{equation}\label{eq: suffices to show1}
E(\nu) \ge \sum_\sigma |w_\sigma| \pi + \Delta(H) \pi.
\end{equation}
Using arguments from \cite{mrz2}, one can show that differentiable
maps are dense in $\Ccal_T(O,S^2) \cap W^{1,2}(O,S^2)$.  Therefore,
we may assume that $\nu$ is differentiable. For each $\sigma$,
choose $s_\sigma \in \Rcal_\nu \cap \Sigma_\sigma$.
Then from Lemma~\ref{lem:1},
it suffices to show that for all $\svec_\sigma \in \Rcal_{\nuvec}\cap \Sigma_\sigma$, we have the inequality
\begin{equation}\label{eq: suffices to show2}
 \sum_{\sigma} \left(D_{\nu}(s_\sigma) - |w_\sigma|\right)
  \ge  \Delta(H).
\end{equation}
Since $D_{\nu}(s_\sigma) \ge |d_\nu(s_\sigma)|$ (cf (\ref{eq: d and
D})) and $d_\nu(s_\sigma) = -w_\sigma $ (cf (\ref{eq: d and w})),
(\ref{eq: suffices to show2}) follows immediately for $H$ conformal
or anticonformal (cf (\ref{eq: Delta conformal})). For $H$
nonconformal, (\ref{eq: suffices to show2}) is equivalent to (cf
(\ref{eq:Delta(H) neq 0}))
\begin{subequations}\label{eq: suffices to show3}
\begin{align}
 \sum_{\sigma} \left(D_{\nu}(s_\sigma) - |w_\sigma|\right)   &\ge  2w_{\sigma_{+}} -
 2\sum_{\sigma\sim \sigma_{+}}
    \Phi(w_\sigma)  - 2\chi,\label{eq: Delta = +}\\
 \sum_{\sigma}\left(D_{\nu}(s_\sigma) - |w_\sigma|\right) &\ge  2 |w_{\sigma_{-}}| -
2\sum_{\sigma\sim \sigma_{-}}
    \Phi(-w_\sigma) - 2\chi.\label{eq: Delta = -}
\end{align}
\end{subequations}

Without loss of generality, we may assume that the edge signs are all equal to $+1$, ie
\begin{equation}\label{eq: edge signs  = 1}
e_x = e_y = e_z = +1.
\end{equation}
(This follows from noting that the Dirichlet energy is invariant
under reflection in, for example, the $xy$-plane of the target
space. That is, if $\nu = (\nu_x,\nu_y, \nu_z)$ and $\nu' =
(\nu_x,\nu_y,-\nu_z)$, then $E(\nu') = E(\nu)$. Under reflection in
the $xy$-plane, the edge signs transform as $(e_x, e_y, e_z) \mapsto
(e_x,e_y,-e_z)$.
Similarly, $e_x$ and $e_y$ change sign under reflections in the
$yz$- and $zx$-coordinate planes respectively.) With (\ref{eq: edge
signs = 1}), the expression (\ref{eq:wrap in terms of e,k,omega})
for the wrapping numbers becomes
\begin{equation}\label{eq: wrap in terms of 2}
   w_\sigma =  \frac{1}{4\pi} \Omega\  + \frac18
+ \frac{1}{2}\sum_j \sigma_j k_j  -\delta_{\sigma,\,(+++)}.
\end{equation}

We proceed to prove (\ref{eq: suffices to show3}).
The argument divides into several
cases according to the signs of the $k_j$'s.  We shall consider  one
representative case in detail, namely where  all the $k_j$'s are positive.  The
arguments for the remaining cases are then briefly sketched.  For definiteness, and without loss of generality,
we  assume that $k_x \le k_y \le k_z$.\\

\noindent{\it Case 1. $k_x, k_y, k_z > 0$.}\quad
In view of (\ref{eq: edge signs  = 1}), we have that $\chi = 0$.
We consider the
bound (\ref{eq: Delta = +}) first. Since $D_\nu(s_\sigma) \ge |w_\sigma|$,
(\ref{eq: Delta = +}) is obviously implied by
\begin{equation}
  \label{eq:Delta_+_restricted}
D_{\nu}(s_{\sigma_+}) -  |w_{\sigma_+}| +
\sum_{ \sigma\sim \sigma_+} \left(D_{\nu}(s_\sigma) -  |w_\sigma|\right)
\ge  2w_{\sigma_{+}} - 2\sum_{\sigma
\sim \sigma_{+}}
    \Phi(w_\sigma),
\end{equation}
in which the sector sum on the left-hand side is restricted to
$\sigma_+$ and the sectors adjacent to $\sigma_+$.  (It turns out that
these are the only sectors in which $D_\nu(s_\sigma)$ is, in certain
cases,
necessarily greater than $|w_\sigma| = |d_\nu(s_\sigma)|$.)
From (\ref{eq: wrap in terms of
2}), we may take $\sigma_+ = (+++)$.
The sectors adjacent to
$\sigma_+$ are  then $(-++)$, $(+-+)$ and $(++-)$.  To simplify the
notation, we replace $(+++)$, $(-++)$, $(+-+)$ and $(++-)$ by
$0$, $1$, $2$ and $3$ respectively.  We let $S = \{s_0, s_1, s_2, s_3\}$, and
 denote a generic point in $S$ by $s_j$.

If we identify $O$ with the unit disk $D^2$, we may identify $\nu$ with an
$S^2$-valued map  $\mu$ on $D^2$.   We proceed to apply
Proposition~\ref{prop:sphere} to obtain a lower bound on $\sum_j
D_\nu(s_j)$.
  For this we  need to calculate
generators for the fundamental group of $S^2 - S$ based at a point
$u$ in the image of the boundary of $\mu$, and to express $[\partial
\mu]$ in terms of them. For definiteness, we take $u = \xhat$
($\xhat$ belongs to the image of $\mu$ since, by assumption, $e_x =
1$).

The loops we consider are sequences of quarter-arcs of great circles
between the coordinate unit vectors $E = \{\pm \xhat$, $\pm \yhat$,
$\pm \zhat\}$.  We will denote these loops as follows.  Given $e, e'
\in E$ with $e \ne -e'$, let $(e, e')$ denote the quarter-arc of the
great circle from $e$ to $e'$ if $e \ne e'$, and the null arc at $e$
if $e = e'$.  Given $e_i \in E$ with $e_i \ne -e_{i+1}$, let $(e_1,
\ldots, e_n)$ denote the curve composed of the sequence $(e_1, e_2),
(e_2,e_3),\ldots, (e_{n-1},e_n)$.  Curves can be concatentated in
the obvious way, ie $(e_1, \ldots, e_m, \ldots, e_n) = (e_1, \ldots,
e_m) (e_m,\ldots, e_n)$.  We let $(e_1,\ldots,e_n)^i$ denote the
curve $(e_1,\ldots, e_n)$ concatenated with itself $i$ times.  In
this notation, the boundary of $\mu$, regarded as a loop in $S^2$
based at $\xhat$, is given by
\begin{equation}\label{eq: boundary of mu as broken arcs}
\partial \mu = C_{z}^{k_z}\, (\xhat, \yhat)\,
C_x^{k_x} \,(\yhat,  \zhat)\,  C_{y}^{k_y} \,(\zhat, \xhat),
\end{equation}
where
\begin{equation}\label{eq: great circles}
C_{x} = (\yhat,\, -\zhat,\,  -\yhat,\,  \zhat,\, \yhat), \ \ C_{y} =
(\zhat,\, -\xhat,\,  -\zhat,\,  \xhat,\, \zhat), \ \ C_z = (\xhat,\,
-\yhat,\, -\xhat,\, \yhat,\, \xhat)
\end{equation}
describe the great circles about the $x$-, $y$- and $z$-axes.

Let
 \begin{equation}\label{eq:tilde gamma}
 \delta_0 = (\xhat,\, \yhat,\, \zhat,\, \xhat),\ \
 \delta_1 = (\xhat,\, \yhat,\, -\xhat,\, \zhat,\, \yhat,\, \xhat),\ \
 \delta_2 = (\xhat,\, \zhat,\, -\yhat,\, \xhat),\ \
 \delta_3 = (\xhat,\, -\zhat,\, \yhat,\, \xhat).
 \end{equation}
It is easily verified that $\delta_j$ is a loop based at $\xhat$
which traverses the boundary of $\Sigma_j$ once with positive
orientation, and therefore separates $s_j$ from the other $s_k$'s
and encloses $s_j$ with positive orientation. Let $c_j \in \pi_1(S^2
- S,\xhat)$ denote the homotopy class of $\delta_j$.

As discussed in Section~\ref{sec: abs degree S2}, $\pi_1(S^2 -
S,\xhat) \cong F(c_1,c_2,c_3)$. Straightforward calculation yields
the following expressions for $c_0$ and $[\partial \mu]$ in terms of
the generators $c_1$, $c_2$ and $c_3$:
\begin{equation}\label{eq:partial_nu}
   c_0 = c_3^{-1} c_1^{-1} c_2^{-1} = (c_2 c_1 c_3)^{-1},\quad
 [\partial \mu] = c_3^{k_z - 1} c_1^{k_x-1} c_2^{k_y-1}.
\end{equation}
Note that, since $k_j >0$ by assumption, the exponents $k_j - 1$ are nonnegative.

Applying Proposition~\ref{prop:sphere} to $\mu$, we get that
\begin{equation}
  \label{eq:Prop_sphere_applied}
\sum_{j=1}^3 D_\mu(s_j) =   \sum_{j = 1}^3 D_{\nu}(s_j) \ge
\min_{g\in \Vcal_{P,N}} \Lambda(g),
\end{equation}
where
\begin{equation}\label{eq: P and N}
P = \half (D_\nu(s_0)  + d_\nu(s_0)), \quad N = \half (D_\nu(s_0)  -
d_\nu(s_0))
\end{equation}
and the minimum is taken over
\begin{equation}
  \label{eq:S_applied}
 g \in  \{c_3^{k_z - 1} c_1^{k_x-1} c_2^{k_y -1}\}
  \langle c_2 c_1 c_3
   \rangle^{P}
   \langle (c_2 c_1 c_3)^{-1}\rangle^{N}.
\end{equation}

The following combinatorial-group-theoretic result implies a bound
on $\Lambda(g)$ in (\ref{eq:Prop_sphere_applied}):
\begin{prop}\label{prop: spelling length}
Let $\Pcal_{p,n} \subset F(A,B,C)$ be the set product given by
\begin{equation}\label{eq: Pcal}
    \Pcal_{p,n} = \langle A^i B^j C^k\rangle  \langle CBA\rangle^p \langle (CBA)^{-1}
\rangle^n.
\end{equation}
Then for $g \in \Pcal_{p,n}$,
\begin{equation}
  \label{eq:spelling_length_result}
  \Lambda(g) \ge i + j + k - (p+n).
\end{equation}
\end{prop}
Thus, for example, the minimum spelling length of words of the form
$f_1(ABC) f_1^{-1} f_2(CBA)^{-1}) f_2^{-1}$, where $f_i \in
F(A,B,C)$, is equal to 2 (we can apply Proposition~\ref{prop: spelling length} to this example by taking $i=j=k=1$ and $p=0, n=1$). This can be seen directly as follows: A
spelling of length $2$ is obtained by taking $f_1 = e$ and $f_2 = A$
to get
 $(ABC) A(A^{-1} B^{-1} C^{-1}) A^{-1} = h_1C h_1^{-1}\, h_2 C^{-1} h_2^{-1}$, where
 $h_1 = AB$ and $h_2 = A$.
By (\ref{eq: spelling mod 2}), the minimum spelling length is either
2 or 0, and a spelling of length zero cannot be found as  $ABC$ and
$CBA$ belong to different conjugacy classes in $F(A,B,C)$. Note that
(\ref{eq: spelling length and degree}) gives the abelian lower bound
\begin{equation}\label{eq: crude lower bound}
  \min_{g \in  \Pcal_{p,n}} \Lambda(g) \ge i + j + k + 3(p - n),
\end{equation}
which for $n \le 2p$ already implies Proposition~\ref{prop: spelling
length}.  Thus, Proposition~\ref{prop: spelling length} is stronger
than the abelian bound (\ref{eq: crude lower bound}) for
\begin{equation}\label{eq: degree lower bound}
    n > 2p
\end{equation}  and therefore, requires independent proof in this case.
 The  proof of
Proposition~\ref{prop: spelling length} is given in the Appendix. In
fact, we believe the following result holds,
\begin{equation}\label{eq: better lower bound}
  \min_{g \in  \Pcal_{p,n}} \Lambda(g) \ge i + j + k - n,
\end{equation}
which is stronger than  Proposition~\ref{prop: spelling length} for
$p > 0$. However, for the purposes of Theorem~\ref{thm:1},
Proposition~\ref{prop: spelling length} is sufficient.

From (\ref{eq:Prop_sphere_applied}), (\ref{eq:S_applied}) and
Proposition~\ref{prop:
  spelling length}, it follows that
\begin{equation}
  \label{eq:sum_of_D's}
  D_\nu(s_0) + \sum_{j=1}^3 D_\nu(s_j) \ge k_x + k_y + k_z - 3,
\end{equation}
where we have used $P + N = D_\nu(s_0)$.
Using (\ref{eq: wrap in terms of 2}), we can express the right-hand
side above
 in terms
of the wrapping numbers,
\begin{equation}
  \label{eq:sum_of_w's}
 k_x + k_y + k_z - 3 = 3w_0 - (w_1 + w_2 + w_3).
\end{equation}
Substituting (\ref{eq:sum_of_w's}) into (\ref{eq:sum_of_D's}) (and
recalling that $w_0 = w_{(+++)} > 0$), we get that
\begin{equation}
  \label{eq:end_of_kpositive}
D_\nu(s_0) - |w_0| +  \sum_{j=1}^3 \left(D_\nu(s_j) - | w_j|\right) \ge 2w_0 -
   \sum_{j=1}^3 (w_j + |w_j|) = 2w_0 -
   2\sum_{j=1}^3 \Phi(w_j),
\end{equation}
which is just the required bound
(\ref{eq: Delta = +}).

The bound (\ref{eq: Delta = -}) is obtained from a similar argument. From
(\ref{eq: wrap in terms of 2}), $\sigma_- = (---)$, so that the sectors adjacent
to $\sigma_-$ are $(+--)$, $(-+-)$, and $(--+)$.
To simplify the
notation, we replace $(---)$, $(+--)$, $(-+-)$, and $(--+)$ by
$0$, $1$, $2$ and $3$ respectively, and let $S = \{s_0, s_1, s_2, s_3\}$.
As above, we argue that
(\ref{eq: Delta = -}) is implied by
\begin{equation}
  \label{eq:Delta_-_restricted}
D_{\nu}(s_{0}) -  |w_{0}| +
\sum_{j = 1}^3 \left(D_{\nu}(s_j) -  |w_j|\right)
\ge  2w_{0} - 2\sum_{j=1}^3
    \Phi(-w_j).
\end{equation}
We proceed to employ Proposition~\ref{prop:sphere} to obtain a lower
bound on $\sum_{j=1}^3 D_\nu(s_j)$.
We introduce loops $\delta_j$ based at $\xhat$,
\begin{equation}\label{eq:tilde gamma2}
\delta_0 = (x\ y\ -x\ -z\ -y\ -x\ y\ x),\ \
\delta_1 = (x\ -y\ -z\ x),\ \
\delta_2 = (x\ -z\ -x\ y\ -z\ x),\ \
\delta_3 = (x\ z\ -x\ -y\ z\ x).
\end{equation}
One can verify that $\delta_j$ separates $s_j$ from the other
$s_k$'s and encloses $s_j$ with positive orientation. Let $c_j \in
\pi_1(S^2 - S,\xhat)$ denote the homotopy classes of $\delta_j$. As
above, we identify $\nu$ with an $S^2$-valued map  $\mu$ on $D^2$.
Calculation gives
\begin{equation}
   c_0 =c_2^{-1} c_1^{-1} c_3^{-1}, \ \
   [\partial \mu] = c_3^{-k_z} c_1^{-k_x} c_2^{-k_y}.
\end{equation}
Then Proposition~\ref{prop:sphere} together with Proposition~
\ref{prop: spelling length}  yield (\ref{eq: Delta = -}).\\

\noindent{Case 2. $k_x, k_y > 0$, $k_z < 0$.}\quad From (\ref{eq:
wrap in terms of 2}), $\sigma_+$ and $\sigma_-$ are given by $(++-)$
and $(--+)$, respectively, and the calculation proceeds as in Case 1
with one key difference.  In the expressions analogous to
(\ref{eq:Prop_sphere_applied}), the spelling length $\Lambda(g)$
is to be minimised over words $g$ of a slightly different form to
that of Proposition~\ref{prop: spelling length} (ie, it is
the conjugacy classes of $ABC$ and its inverse which appear, rather
than those of $CBA$).  By analogy with Proposition~\ref{prop:
spelling length}, we have the following:
\begin{prop}\label{prop: spelling length 2}
Let $\Qcal_{p,n} \subset F(A,B,C)$ be the set product given by
\begin{equation}\label{eq: Qcal}
    \Qcal_{p,n} = \langle A^i B^j C^k\rangle  \langle ABC\rangle^p \langle (ABC)^{-1}
\rangle^n.
\end{equation}
Then for $g \in \Qcal_{p,n}$,
\begin{equation}
  \label{eq:spelling_length_result_2}
  \Lambda(g) \ge i + j + k - (p+n +2).
\end{equation}
\end{prop}
\noindent Thus, for example, the minimum spelling length of words of
the form $ABC h C^{-1} B^{-1} A^{-1} h^{-1}$ for $h \in F(A,B,C)$,
is obviously zero. The contribution  $-2$ in
(\ref{eq:spelling_length_result_2})
 accounts for
the term $\chi = 1$ in $\Delta(H)$ in this case (cf (\ref{eq: chi}) and
(\ref{eq:Delta(H) neq 0})) .\\

\noindent{Case 3. $k_x > 0$, $k_y, k_z < 0$.}\quad From (\ref{eq:
wrap in terms of 2}), $\sigma_+$ and $\sigma_-$ are given by $(+--)$
and $(-++)$
respectively, and the calculation proceeds as in Case 1.\\

\noindent{Case 4. $k_x, k_y, k_z < 0$.}\quad From (\ref{eq: wrap in
terms of 2}), $\sigma_+$ and $\sigma_-$ are given by $(---)$ and
$(+++)$
respectively, and the calculation proceeds as in Case 2.\\

\noindent{Case 5. $k_x k_y k_z = 0$}.\quad  If one of the $k_j$'s
vanishes, then in most cases, $\sigma_+$ and one of its adjacent
sectors will share the same largest wrapping number, and similarly
for $\sigma_-$.  Then $\Delta(H) = 0$, and (\ref{eq: suffices to
show2}) follows automatically.  Exceptions can arise when $k_x =
k_y$ and $k_z = 0$, but it is straightforward to verify (\ref{eq:
suffices to show3}) in this special case \cite{thesis}; the argument
is omitted.

\end{proof}

\section{Upper bound for $\Ecal(H)$} \label{sec: upper bound}

Given a homotopy class $H\subset \Ccal_T\left(O,S^2\right)$, we can construct an explicit representative $\nuvec\in H$ that realizes the upper bound in Theorem~\ref{thm:2}.
As stated in Section~\ref{sec: statement}, homotopy classes in $\Ccal_T\left(O,S^2\right)$ are classified as being either conformal, anticonformal or nonconformal. Conformal and anticonformal homotopy classes are studied in detail in \cite{mrz2, mrz3}. For these homotopy classes, $\Delta(H)=0$ by definition (see (\ref{eq:Delta(H) neq 0}))
and consequently, the infimum Dirichlet energy is bounded from below by $\Ecal(H) \geq \sum_\sigma |w_\sigma|$ i.e. the lower bound is simply the abelian bound in 
(\ref{eq: spelling length and degree}). In \cite{mrz2}, we construct explicit conformal (anticonformal) representatives $\nuvec$ for conformal (anticonformal) homotopy classes such that for a regular value $\svec_\sigma \in \Sigma_\sigma \cap \Rcal_\nuvec$, the absolute degree coincides with the absolute value of the wrapping number i.e. $D_\nuvec (\svec_\sigma) = |w_\sigma|$ for all $\sigma$ and the corresponding Dirichlet energy is 
\begin{equation}
\label{eq:ub1}
E(\nuvec) = \pi \sum_\sigma D_\nuvec \left(\svec_\sigma \right) = \pi \sum_\sigma |w_\sigma|.
\end{equation}
Thus, these representatives achieve the upper bound in Theorem~\ref{thm:2}.

Nonconformal homotopy classes have also been studied in some detail in \cite{mrz3}. In \cite{mrz3}, we construct explicit representatives $\nuvec$ in nonconformal homotopy classes from a juxtaposition of conformal and anticonformal configurations. The representative $\nuvec$ is taken to be either conformal or anticonformal almost everywhere in $O$ except for a small interior disc. We insert a certain number, $N\geq 1$, of full coverings of $S^2$ with either positive or negative orientation within this interior disc. The choice of $N$ and the orientation of these full coverings (positive or negative) clearly depends on the nonconformal homotopy class in question. One can show that the representative $\nuvec$, thus defined, has Dirichlet energy strictly greater than the upper bound in Theorem~\ref{thm:2}.

In this section, we return to the upper bound problem for nonconformal homotopy classes. We construct alternative explicit representatives by introducing quarter-sphere configurations. We take the representative $\nuvec$ to be either conformal or anticonformal everywhere away from the vertices of $O$. Near the vertices of $O$, we modify $\nuvec$ and insert quarter-sphere configurations. The quarter-sphere configurations are either conformal or anticonformal configurations that cover a pair of adjacent octants with either negative or positive orientation and preserve the tangent boundary conditions. The quarter-sphere configurations allow us to realize the minimal number of pre-images consistent with (\ref{eq: lower bound}) and hence, saturate the lower bound in Theorem~\ref{thm:1} and realize the upper bound in Theorem~\ref{thm:2}. In Section~\ref{sec:seg}, we consider a simple illustrative example.
In Section~\ref{sec:qsphere}, we review the main results
for conformal and anticonformal topologies from \cite{mrz2} and
formally define quarter-sphere configurations. In
Section~\ref{sec:rep}, we explicitly define the representatives
for nonconformal homotopy classes and in Section~\ref{sec:ub}, we
carry out the relevant energy estimates that yield a proof for Theorem~\ref{thm:2}. 

\subsection{An Example} \label{sec:seg}

Let $H$ be the nonconformal homotopy class defined by the invariants
\begin{eqnarray}
\label{eq:ub2}
&& e_j = +1, \quad k_j = +1, \quad \forall j \nonumber \\
&&  \Omega = \frac{3\pi}{2} .
\end{eqnarray}
The corresponding wrapping numbers are shown below (\ref{eq:wrap in terms of e,k,omega}) -
\begin{eqnarray}
\label{eq:ub3}
&& w_{+++} = w_{-++} = w_{+-+} = w_{++-} = 1 \nonumber \\
&& w_{+--} = w_{-+-} = w_{--+} = 0 \nonumber \\
&& w_{---} = -1.
\end{eqnarray} From Theorem~\ref{thm:1}, $\Ecal(H) \geq 7\pi $ and $\Delta(H) = 2$ so that there necessarily exists an octant $\Sigma_\sigma$, with $\sigma \sim \left(---\right),$ such that a regular value $\svec_\sigma \in \Sigma_\sigma$ has at least two pre-images in spite of the fact that $w_\sigma = 0$ (or equivalently $d_\nuvec(\svec_\sigma) = 0$).

We construct a representative $\nuvec\in H$ on the following lines. Let $O_z$ denote a small neighbourhood of the vertex $\zhat$ in $O$.
We take $\nuvec$ to be an anticonformal configuration on $O\setminus O_z$ with wrapping numbers as shown below
\begin{eqnarray}
\label{eq:ub7}
&& w_{+++}= w_{-++} = w_{+-+} = w_{++-}  = w_{--+} = 1 \nonumber \\
&& w_{\sigma} = 0 \quad \textrm{otherwise.}
\end{eqnarray} In $O_z$, we modify $\nuvec$ to insert a quarter-sphere configuration. This quarter-sphere configuration is a conformal configuration 
by construction and its image covers the pair of adjacent octants,
 $\Sigma_{_{--+}}$ and $\Sigma_{_{---}}$, exactly once with negative orientation. The representative, $\nuvec$, thus defined from the juxtaposition of the anticonformal configuration in (\ref{eq:ub7}) and the quarter-sphere configuration, has the correct topology $H$ in (\ref{eq:ub2}), satisfies the tangent boundary conditions and is continuous everywhere in $O$. The corresponding absolute degrees for regular values $\svec_\sigma \in \Rcal_\nuvec\cap \Sigma_\sigma$ are
\begin{eqnarray}
\label{eq:ub8}
&& D_\nuvec(\svec_\sigma) = 1 \quad \textrm{ for $\sigma=\left\{ (+++), (-++), (+-+), (++-)\right\}$} \nonumber \\
&& D_\nuvec(\svec_\sigma) = 2 \quad \textrm{for $\sigma =  (--+)$} \nonumber \\
&&  D_\nuvec(\svec_\sigma) = 1 \quad \textrm{ for $\sigma = (- - -)$.}
\end{eqnarray} Since $\nuvec$ is either conformal or anticonformal almost everywhere by construction, we can explicitly estimate its Dirichlet energy from (\ref{eq:ub1})
\begin{equation}
\label{eq:ub9}
\Ecal(H) \leq E(\nuvec) = \pi \sum_\sigma D_\nuvec(\svec_\sigma) = 7\pi,
\end{equation} consistent with the upper bound in Theorem~\ref{thm:2}.

For the sake of comparison, we briefly outline the construction of a representative $\nuvec\in H$, following the methods in \cite{mrz3}.
Let $D_\eps$ denote a small interior disc of radius $0<\eps <\frac{1}{8}$. In \cite{mrz3}, we take $\nuvec$ to be an anticonformal configuration on $O\setminus D_\eps$ with wrapping numbers
-
\begin{eqnarray}
\label{eq:ub4}
&& w_{+++} = w_{-++} =  w_{+-+} =  w_{++-} = 2 \nonumber \\
&& w_{+--} = w_{-+-} = w_{--+} = 1 \nonumber \\
&& w_{---} = 0
\end{eqnarray} and in $D_\eps$, we insert a conformal configuration that covers $S^2$ exactly once with negative orientation. The absolute degrees in this case are given by 
\begin{eqnarray}
\label{eq:ub5}
&& D_\nuvec(\svec_\sigma) = 3 \quad \sigma=\left\{ (+++), (-++), (+-+), (++-)\right\} \nonumber \\
&& D_\nuvec(\svec_\sigma) = 2 \quad \sigma = \left\{ (+--), (-+-), (--+)\right\} \nonumber \\
&&  D_\nuvec(\svec_\sigma) = 1 \quad \sigma = (- - -) \quad \textrm{where $\svec_\sigma \in \Sigma_\sigma \cap \Rcal_{\nuvec}$.}
\end{eqnarray}
The corresponding Dirichlet energy is
\begin{equation}
\label{eq:ub6}
\Ecal(H) \leq E(\nuvec) = \pi \sum_\sigma D_\nuvec(\svec_\sigma) = 19\pi,
\end{equation} which is more than twice the upper bound in Theorem~\ref{thm:2}.

This example demonstrates how a quarter-sphere configuration can enable us to realize the upper bound in Theorem~\ref{thm:2} and realize the minimum number of pre-images consistent with Theorem~\ref{thm:1}. More generally, we need a sequence of alternating conformal and anticonformal quarter-sphere configurations localized near the vertices of $O$ and these quarter-sphere configurations are chosen carefully so as to preserve the topology and the tangent boundary conditions. Explicit details are given in the subsequent sections.

\subsection{Complex representation}
\label{sec:qsphere} 

For a given homotopy class $H\subset \Ccal_T(O,S^2)$, we represent the representative $\nuvec\in H$ by a complex-valued function using stereographic projection.
Let $\mathbb{C}^*
= \Cc \cup \left\{ \infty \right\}$ denote the extended complex
plane. Let $P:S^2 \to \Cc^*$ denote the stereographic projection
of the unit sphere into the extended complex plane with $-\zhat$
being projected to $\infty$ i.e. for $\evec = (e_x, e_y, e_z) \in
S^2$, $P(\evec)$ is given by
\begin{equation}
\label{eq:ub10}
P(\evec) = \frac{e_x + \ci e_y}{1 + e_z}.
\end{equation}

We let $Q = P(O)\subset \Cc^*$ denote the projected domain; then $Q$ is the quarter-disc given by
\begin{equation}
\label{eq:ub11} Q = \left\{  w = \rho \exp^{\ci \phi} | 0\leq \rho
\leq 1,~ 0\leq \phi \leq \frac{\pi}{2} \right\}.
\end{equation} The boundary of $Q$ consists of three segments, $\partial Q = C_1 \cup C_2 \cup C_3$  - (i) the real segment $C_1 =\left\{ w \in \Rr; 0\leq w\leq 1 \right\}$, which is the projection of the $zx$-edge of $O$, (ii) the imaginary segment $C_2 = \left\{ w = \ci t; 0\leq t \leq 1\right\}$, which is the projection of the $yz$-edge of $O$ and (iii) the quarter-circle $C_3 = \left\{ w = \exp^{\ci \phi}; 0\leq \phi \leq \frac{\pi}{2}\right\}$, which is the projection of the $xy$-edge of $O$. The vertices of $Q$ are at the points $P(\xhat) = 1$, $P(\yhat) = \ci$ and $P(\zhat) = 0$ respectively.

Given $\nuvec:O\to S^2$, we define the corresponding projected map $K:Q \to \Cc^*$ by
\begin{equation}
\label{eq:ub12}
K = P \circ \nuvec \circ P^{-1}.
\end{equation}
Then if $\nuvec = \left(\nu_x, \nu_y, \nu_z \right)$, we have that
\begin{equation}
K(w) = \frac{\nu_x(\svec) + \ci \nu_y(\svec)}{1 + \nu_z(\svec)} \quad \textrm{where $\svec = P^{-1}(w)$.}
\label{eq:ub13}
\end{equation} Let $\Ccal_T\left(Q, \Cc^*\right)$ denote the space of maps $K:Q\to \Cc^*$ for which $\nuvec \in \Ccal_T(O,S^2)$. The tangent boundary conditions require that (i) $K(w)$ is real if $w$ is real (i.e. $w\in C_1$) (ii) $K(w)$ is imaginary if $w$ is imaginary (for $w \in C_2$) and (iii) $|K(w)|=1$ if $|w|=1$ (for $w\in C_3$). Finally, if $\nuvec$ is differentiable, then so is $K$ and the Dirichlet energy of $\nuvec$ is given in terms of $K$ as shown below -
\begin{equation}
\label{eq:ub14}
E(\nuvec) = E(K) = \int_{Q}  \mathcal{H} (K)~d^2 w
\end{equation} where
\begin{equation}
\label{eq:ub15}
\mathcal{H} (K) =  4 \left( \frac{|\partial _w K|^2 + |\partial_{\wbar} K |^2}{\left(1 + |K|^2\right)^2} \right)
\end{equation}
is the Dirichlet energy density in complex coordinates.

\subsubsection{Conformal and anticonformal representatives}
\label{sec:conf}

We briefly review the main results in \cite{mrz2, mrz3} for conformal and anticonformal homotopy classes. In \cite{mrz3}, we show that a homotopy class $H\subset \Ccal_T(O,S^2)$ is conformal (anticonformal) if and only if it admits a conformal (anticonformal) representative. Under the stereographic projection $P:S^2 \to \Cc^*$ defined in (\ref{eq:ub10}), we represent a conformal representative in a conformal homotopy class by an analytic function $f:Q \to \Cc^*$.  The tangent boundary conditions means that if $w$ is a zero of $f$, then so are $\pm \wbar$ and $-w$ and $\frac{1}{w}$ is a pole. These constraints along with the conditions that $f(0) = 0$ or $f(0) = \infty$ (since $\nuvec(\zhat) = \pm \zhat$) and $f(1) = \pm 1$ (since $\nuvec(\xhat) = \pm \xhat$) imply that $f(w)$ is a rational function of the form \cite{mrz2}
\begin{multline}
    \label{eq:f}
    f(w) = \pm w^{2m+1}
  \prod_{j=1}^a\left(
  \frac{w^2 - r_j^2}{r_j^2 w^2 - 1}
  \right)^{\rho_j}
  \prod_{k=1}^b\left(
  \frac{w^2 + s_k^2}{s_k^2 w^2 + 1}
  \right)^{\sigma_k}\times\\
\times  \prod_{l=1}^c\left(
  \frac{(w^2 - t_l^2)(w^2 -  \tbar_l^2)}
  {(t^2_l w^2 - 1)({\tbar}_l^2w^2  - 1)}
  \right)^{\tau_l}.
\end{multline}
The $r_j$'s denote the real zeros ($\rho_j = 1$) and poles ($\rho_j =
-1)$ of $f$ between $0$ and $1$; the $s_k$'s, the imaginary zeros
and poles of $f$ (according to whether $\sigma_k = \pm 1$) between $0$ and $i$
; and the $t_l$'s, the complex zeros and poles of $f$ (according to whether
$\tau_l = \pm 1$) with modulus less than one and argument between 0 and
$\pi/2$, $m$ is any integer and $a$, $b$ and $c$ are non-negative integers.The homotopy invariants $(e, k,\Omega)$ can be explicitly computed in terms of the parameters $\left\{m, a, b, c, r_j, s_k, \tau_l, \rho_j, \sigma_k, \tau_l\right\}$.
Analogous remarks apply to anticonformal representatives, with complex analytic functions being replaced by complex antianalytic functions $f(\wbar)$, where $f$ is rational and of the form (\ref{eq:f}).

For $H = \left(e, k,\Omega \right)$ conformal or anticonformal, let $F_H$ denote the corresponding conformal/anticonformal representative of the form (\ref{eq:f}). For such representatives $F_H$, $D_{F_H}(\svec_\sigma)$ is independent of the choice of the regular value $\svec_\sigma \in \Sigma_\sigma$ and $D_{F_H}(\svec_\sigma) = |w_\sigma|$ for all $\sigma$. The corresponding Dirichlet energy can be explicitly computed as in (\ref{eq:ub1}) and  
\begin{equation}
\label{eq:ub17} E\left(F_H\right) = \pi \sum_\sigma
D_{F_H}\left(\sigma \right) = \pi \sum_\sigma |w_\sigma |.
\end{equation}
Therefore, $E(F_H)$ realizes the lower bound of Theorem~\ref{thm:1} and the upper bound in Theorem~\ref{thm:2} i.e.
\begin{equation}
\label{eq:ub18}
E\left(F_H\right) = \Ecal (H)
\end{equation}
as required.

\subsubsection{Quarter-sphere configurations}
\label{sec:qsconf}

Quarter-sphere configurations are defined on small neighbourhoods
of the vertices of $Q$. For concreteness, let $$Q_\eps = \left\{ w \in Q; \rho = |w|\leq \eps
\right\}\subset Q$$ denote the closed $\eps$-neighbourhood of the vertex $w = 0$
where $\eps>0$. Let
\begin{equation}
\label{eq:ub19}
\rho_0 = 0~ \textrm{and} ~\rho_m = \eps^{L+1 - m}~\textrm{with $1 \leq m \leq L$,}
\end{equation} where $L$ is a positive integer. We note that
$$\frac{\rho_{m+1}}{\rho_m} = \frac{1}{\eps}. $$

We partition $Q_\eps$ into $(2L-1)$ concentric quarter-annuli -
(i) the quarter-sphere configurations are defined on the annuli
\begin{equation}
\label{eq:ub20}
2\rho_{m-1} \leq \rho  \leq \rho_m \quad 1\leq m \leq L
\end{equation} and (ii) we interpolate between the different quarter-sphere
configurations on the intervening annuli
\begin{equation}
\label{eq:ub21}
\rho_{n} \leq \rho  \leq 2\rho_n \quad 1\leq n \leq L-1.
\end{equation}

Consider the quarter-annuli for $1\leq m \leq L$. We define the quarter-sphere configurations
$g_{m,\eps}$ as shown below -
\begin{equation}
  \label{eq:ub22}
  g_{m,\eps}(w) =
  \begin{cases}
    -\frac{w}{\sqrt{\eps} \rho_m}, & \textrm{$2\rho_{m-1} \leq \rho \leq \rho_m$, $m$ ~  odd},\\
 \frac{\rho_{m-1}}{\sqrt{\eps} \wbar},&
   \textrm{$2\rho_{m-1} \leq \rho \leq \rho_m$, $m$ ~ even}.
  \end{cases}
\end{equation}
We consider the case of $m$ odd first. For $m$
odd, $g_{m,\eps}$ is a rational analytic function (conformal configuration) which is
real on the real axis and imaginary on the
imaginary axis i.e. $g_{m,\eps}$ satisfies the tangent boundary
conditions on the real and imaginary axes as required. One can directly verify that the image of
$g_{m,\eps}$ covers the pair of adjacent octants
$\Sigma_{_{--\pm}}$ exactly once with
negative orientation, except for a small neighbourhood of $\pm \zhat$ on $S^2$ \begin{footnote}{Strictly speaking, the image of $g_{m,\eps}$ covers the projected octants $P(\Sigma_{--\pm})$ on $\Cc^*$ except for a $\sqrt{\eps}$-neighbourhood of $w=0$ and $\infty$ but here and in what follows, we do not explicitly distinguish between the octants and their projection on the complex plane.}\end{footnote}. For regular values $\xi_\sigma$ not contained in these excluded neighbourhoods i.e. for $\xi_\sigma \in \Rcal_{g_{m,\eps}}$ satisfying
\begin{equation}
\label{eq:regular}
\left|\xi_\sigma \right|, \frac{1}{\left|\xi_\sigma \right|} < \frac{1}{\sqrt{\eps}},
\end{equation}
we have that
\begin{equation}
  \label{eq:ub23}
  d_{g_{m,\eps}}(\xi_\sigma) =
  \begin{cases}
    -1, & \sigma = \left(- - \pm \right),\\
 0,&
    otherwise
  \end{cases}
\end{equation} where $d_{g_{m,\eps}}(\xi_\sigma)$ is the algebraic degree defined in (\ref{eq: d_f}). The corresponding Dirichlet energy is easily estimated using (\ref{eq:ub14}) and (\ref{eq:ub15}) and we have the following - 
\begin{equation}
\label{eq:ub24} E\left(g_{m,\eps}\right) = \int_{2\rho_{m-1} \leq
\rho \leq \rho_m} \mathcal{H}\left(g_{m,\eps}\right)~d^2 w  =
2\pi\left( \frac{1 - 4\eps^2}{(1+\eps)(1+4\eps) } \right)
\end{equation}
and
\begin{equation}
\label{eq:ub25}
E\left(g_{m,\eps}\right) \leq 2\pi + C\eps
\end{equation} where $C$ is a positive constant independent of $\eps$.

The case of $m$ even can be treated in an analogous manner. Here, $g_{m,\eps}$ is a rational antianalytic function (anticonformal configuration) that is real on the real axis and imaginary on the imaginary axis. Again, one can directly verify that the image of $g_{m,\eps}$ covers the pair of adjacent octants $\Sigma_{_{++\pm}}$ exactly once with positive orientation except for a small neighbourhood of $\pm\zhat$ on $S^2$ and for regular values $\xi_\sigma$ satisfying (\ref{eq:regular}),  the algebraic degrees are given by
\begin{equation}
  \label{eq:ub26}
  d_{g_{m,\eps}}(\xi_\sigma) =
  \begin{cases}
    +1, & \sigma = \left(+ + \pm \right),\\
 0,&
    otherwise.
  \end{cases}
\end{equation}The corresponding Dirichlet energy is estimated as in (\ref{eq:ub24}) and we have that
\begin{equation}
\label{eq:ub27}
E\left(g_{m,\eps}\right) = \int_{2\rho_{m-1} \leq \rho \leq \rho_m} \mathcal{H}\left(g_{m,\eps}\right)~d^2 w \leq  2\pi + C\eps
\end{equation} for a positive constant $C$ independent of $\eps$.

On the annuli $\rho_n \leq \rho \leq 2\rho_n$ with $1\leq n \leq L-1$, we define the interpolatory functions $h_{n,\eps}$ according to
\begin{equation}
  \label{eq:ub28}
  h_{n,\eps}(w) =
  \begin{cases}
   \left((1 - s_n(\rho)) / g_{n,\eps} (w) + s_n(\rho) /g_{n+1,\eps}(w)\right)^{-1}, & \textrm{$n$ ~  odd},\\
 (1 - s_n(\rho)) g_{n,\eps} + s_n(\rho) g_{n+1,\eps}(w),&
    \textrm{$n$ ~ even}
  \end{cases}
\end{equation} where $s_n$ is the switching function
\begin{equation}
\label{eq:ub29}
s_n(\rho) = \frac{\rho - \rho_n}{\rho_n}, \quad \rho_n \leq \rho \leq 2\rho_n.
\end{equation} It is easy to verify that $h_{n,\eps}$, thus defined, is real on the real axis and imaginary on the imaginary axis since the $g_{n,\eps}$'s satisfy tangent boundary conditions on the real and imaginary axes. The energy estimates for $h_{n,\eps}$ can be easily carried out.
We consider the case of $n$ odd in (\ref{eq:ub28}) first  with $g_{n,\eps}=-\frac{w}{\sqrt{\eps} \rho_n}$, $g_{n+1,\eps} = \frac{\rho_n}{\sqrt{\eps} \wbar} $. For $s_n$ in (\ref{eq:ub29}), we have that
\begin{equation}
\label{eq:new1}|\partial_w s_n|^2 + |\partial_{\wbar} s_n|^2 < \frac{1}{\rho_n^2}.\end{equation}
Similarly, we note that
\begin{equation}
\label{eq:new2}\left|\partial_w \frac{1}{h_{n,\eps}}\right|^2 \leq C \left(\left|\partial_w \left(\frac{1}{g_{n,\eps}}\right)\right|^2 + |\partial_w s_n|^2\left(\left|\frac{1}{g_{n,\eps}}\right|^2 + \left|\frac{1}{g_{n+1,\eps}}\right|^2 \right)\right) \leq C^{'}\frac{\eps}{\rho_n^2}\end{equation}
and likewise
\begin{equation}
\label{eq:new3}\left|\partial_{\wbar} \frac{1}{h_{n,\eps}}\right|^2 \leq C \left(\left|\partial_{\wbar} \left(\frac{1}{g_{n+1,\eps}}\right)\right|^2 + |\partial_{\wbar} s_n|^2\left(\left|\frac{1}{g_{n,\eps}}\right|^2 + \left|\frac{1}{g_{n+1,\eps}}\right|^2 \right)\right) \leq C^{''}\frac{\eps}{\rho_n^2}\end{equation}
where $C$, $C^{'}$ and $C^{''}$ are positive constants independent of $\eps$. On the other hand
\begin{equation}
\label{eq:new4}\left(1 + \left|\frac{1}{h_{n,\eps}}\right|^2\right) \geq 1\end{equation}
so that
\begin{equation}
\label{eq:new5}\Hcal\left(\frac{1}{h_{n,\eps}}\right) \leq D \frac{\eps}{\rho_n^2}\end{equation}
from (\ref{eq:ub15}), for a positive constant $D$ independent of $\eps$. Substituting the above into (\ref{eq:ub14}), we obtain the following
- 
\begin{equation}
\label{eq:ub32new}
E\left(\frac{1}{h_{n,\eps}}\right) \leq C_1 ~ \eps, \quad \textrm{$n$~odd}
\end{equation}
for a positive constant $C_1$ independent of $\eps$ and since $\Hcal\left(\frac{1}{h_{n,\eps}}\right) = \Hcal\left(h_{n,\eps}\right)$, we have that
\begin{equation}\label{eq:ub32}
E\left(h_{n,\eps}\right) \leq C_1 ~ \eps.
\end{equation}
We repeat the same calculations for the case $n$ even and it can be shown that the energy estimate (\ref{eq:ub32}) holds for all $n$ i.e.
\begin{equation}
\label{eq:ub34}
E\left(h_{n,\eps}\right) \leq C_2 ~ \eps, \quad 1\leq n \leq L-1
\end{equation} where $C_2$ is a positive constant independent of $\eps$.

From Lemma~\ref{lem:1}, we have that the Dirichlet energy of $h_{n,\eps}$ is bounded from below by
\begin{equation}
\label{eq:ub35}
E\left(h_{n,\eps}\right) \geq
 2 \sum_\sigma \int_{\xi_\sigma \in \Rcal_{h_{n,\eps}} \cap P(\Sigma_\sigma)}
D_{h_{n,\eps}}(\xi_\sigma)  \, d^2 w.
\end{equation}
Since $E\left(h_{n,\eps}\right) \to 0$ as $\eps \to 0$ and $D_{h_{n,\eps}}(\xi_\sigma) \geq |d_{h_{n,\eps}}(\xi_\sigma)|$, we deduce that
\begin{equation}
\label{eq:ub36}
d_{h_{n,\eps}}(\xi_\sigma) = 0 \quad  ~1\leq n \leq L-1.
\end{equation}
for all regular values $\xi_\sigma$ satisfying (\ref{eq:regular}), apart from a subset whose measure vanishes with $\epsilon$ and for all $\sigma$.

We define the configuration $\Gamma_{L,\eps}:Q_\eps \to \Cc^*$ as follows -
\begin{equation}
  \label{eq:ub37}
  \Gamma_{L,\eps}(w) =
  \begin{cases}
  g_{m,\eps}(w), & 2\rho_{m-1} \leq \rho \leq \rho_m,\\
 h_{n,\eps}(w),& \rho_{n} \leq \rho \leq 2 \rho_n
  \end{cases}
\end{equation} where $1\leq m \leq L$ and $1\leq n \leq L-1$. 

\begin{prop}
\label{prop:ub1}
The function $\Gamma_{L,\eps}:Q_\eps \to \Cc^*$ defined in (\ref{eq:ub37}) has the following properties -

(i) $\Gamma_{L,\eps}$ is real (imaginary) on the real (imaginary) axis i.e. $\Gamma_{L,\eps}$ satisfies the tangent boundary conditions on the real and imaginary axes,

(ii) the Dirichlet energy is bounded from above by
\begin{equation}
E(\Gamma_{L,\eps})  \leq 2\pi L + C_3 \eps
\label{eq:ub39}
\end{equation} where $C_3$ is a positive constant independent of $\eps$, 

(iii) the algebraic degrees are given by
\begin{equation}
\label{eq:ub40}
d_{\Gamma_{L,\eps}}\left(\xi_\sigma \right) = W_\sigma (L)
\end{equation} for regular values $\xi_\sigma$ satisfying (\ref{eq:regular}) and $W_\sigma(L)$ is defined as shown below
\begin{equation}
  \label{eq:ub38}
 W_\sigma (L) =
  \begin{cases}
 -\left[\frac{L+1}{2}\right] ,& \sigma=\left(- - \pm \right),\\
 \left[\frac{L}{2}\right],& \sigma=\left(+ + \pm \right),\\
0, & otherwise.
  \end{cases}
\end{equation}
\end{prop}

\begin{proof}
Property (i) is immediate from the definition of $\Gamma_{L,\eps}$ in (\ref{eq:ub37}), since $g_{m,\eps}$ and $h_{n,\eps}$ satisfy the tangent boundary conditions on the real and imaginary axes. To show (\ref{eq:ub39}), we note that
\begin{equation}
\label{eq:ub41} \int_{Q_\eps}
\mathcal{H}\left(\Gamma_{L,\eps}\right)~d^2 w = \sum_{m=1}^{L}
\int_{2\rho_{m-1}\leq \rho \leq \rho_m}
\mathcal{H}\left(g_{m,\eps}\right)~d^2 w +  \sum_{n=1}^{L-1}
\int_{\rho_{n}\leq \rho \leq 2\rho_n}
\mathcal{H}\left(h_{n,\eps}\right)~d^2 w,
\end{equation} and substitute the energy estimates (\ref{eq:ub25}), (\ref{eq:ub27}) and (\ref{eq:ub34}) into (\ref{eq:ub41}), yielding the upper bound (\ref{eq:ub39}).

To show (\ref{eq:ub40}), we observe that for $L$ even, we have an equal number of conformal and anticonformal quarter-sphere configurations whereas for $L$ odd, we have $(L+1)/2$ conformal quarter-sphere configurations and $(L-1)/2$ anticonformal quarter-sphere configurations (refer to (\ref{eq:ub22})). The algebraic degree, $d_{\Gamma_{L,\eps}} \left(\xi_\sigma \right)$, of a regular value  $\xi_\sigma$ is given by 
\begin{equation}
\label{eq:ub42}
d_{\Gamma_{L,\eps}}\left(\xi_\sigma \right) =
\sum_{m=1}^{L}d_{g_{m,\eps}}\left(\xi_\sigma \right) + \sum_{n=1}^{L-1}d_{h_{n,\eps}}\left(\xi_\sigma \right).
\end{equation}
It suffices to substitute
(\ref{eq:ub23}), (\ref{eq:ub26}) and (\ref{eq:ub36}) into (\ref{eq:ub42}) and (\ref{eq:ub40}) immediately follows.
\end{proof}

\textit{Note:} We note that $d_{\Gamma_{L,\eps}}\left(\xi_\sigma \right)$ is independent of the choice of regular value $\xi_\sigma$ and only depends on $\sigma$; therefore, these algebraic degrees are referred to as $d_{\Gamma_{L,\eps}}(\sigma)$ in the subsequent sections. 

\subsubsection{Symmetries and translations of quarter-spheres}

 Let $\gamma: \Cc^* \rightarrow \Cc^*$ denote the M\"obius
transformation
\begin{equation}\label{eq:ub45}
    \gamma(w) = \frac{i-w}{i+w}.
\end{equation}
It is easy to verify that $\gamma(1) = i$, $\gamma(i)
= 0$, $\gamma(0) = 1$ and $\gamma(Q)=Q$.  We define
\begin{equation}\label{eq:ub46}
    \gamma_x = \gamma, \quad \gamma_y = \gamma^2, \quad \gamma_z = \gamma^3(\text {$=$ id}).
\end{equation}
Then
\begin{equation}\label{eq:ub47}
    Q_{j,\epsilon} = \gamma_j(Q_\epsilon)
\end{equation}
is the closed $\epsilon$-neighbourhood of $\gamma_j(0)$ in $Q$; for example, $Q_{x,\eps}$ is a closed $\eps$-neighbourhood of the vertex $w = 1$, $Q_{y,\eps}$ is a closed $\eps$-neighbourhood of the vertex $w = \ci$ and $Q_{z,\eps} = Q_{\eps}$.

Let $M_j\geq 1$ be a positive integer. We define $G_{j,M_j,\epsilon}: Q_{j,\epsilon} \rightarrow \Cc^*$ by
\begin{equation}\label{eq:ub48}
G_{j,M_j,\epsilon} = \gamma_j \circ \Gamma_{M_j,\epsilon} \circ
\gamma_j^{-1}
\end{equation} For $\sigma = (\sigma_x\, \sigma_y\, \sigma_z)$, we define the following permutations
 \begin{equation}\label{eq:ub44}
p_x(\sigma) = (\sigma_z\, \sigma_x\, \sigma_y) , ~ p_y(\sigma) = (\sigma_y\, \sigma_z\, \sigma_x)~ \textrm{and $p_z(\sigma) = (\sigma_x\, \sigma_y\, \sigma_z)$.}
\end{equation} Then for $G_{x,M_x,\eps}$, the octants $\Sigma_{_{--\pm}}$ and $\Sigma_{_{+ + \pm}}$ are mapped onto $\Sigma_{_{p_x\left(--\pm\right)}} = \Sigma_{_{\pm - -}}$ and $\Sigma_{_{p_x\left(+ + \pm\right)}} = \Sigma_{_{\pm + +}}$ respectively. Therefore, the conformal quarter-sphere configurations in $G_{x,M_x,\eps}$ ($m$ odd in (\ref{eq:ub22})) cover the octant pair $\Sigma_{_{\pm - -}}$ exactly once with negative orientation and the anticonformal quarter-sphere configurations ($m$ even in (\ref{eq:ub22})) cover the octant pair $\Sigma_{_{\pm + +}}$ exactly once with positive orientation. Similarly, for $G_{y,M_y,\eps}$, the octants $\Sigma_{_{--\pm}}$ and $\Sigma_{_{+ + \pm}}$ are mapped onto $\Sigma_{_{p_y\left(- - \pm\right)}} = \Sigma_{_{- \pm -}}$ and $\Sigma_{_{p_y\left(+ + \pm\right)}} = \Sigma_{_{+ \pm +}}$ respectively. Therefore, the conformal quarter-sphere configurations cover the octant pair $\Sigma_{_{- \pm -}}$ exactly once with negative orientation and the anticonformal quarter-sphere configurations cover the octant pair $\Sigma_{_{+ \pm +}}$ exactly once with positive orientation.

It then follows directly from Proposition~\ref{prop:ub1} that
 \begin{eqnarray}
 && \textrm{ $G_{j,M_j,\epsilon}(w)$ satisfies the tangent boundary conditions on $Q_{j,\epsilon}\cap \partial Q$}, \label{eq:ub49}\\
 && \int_{Q_{j,\epsilon}} \Hcal(G_{j,M_j,\epsilon}) \,d^2 w \leq  2 M_j\pi + C \epsilon ,\label{eq:ub50}\\
&& d_{G_{j,M_j,\epsilon}}(\xi_\sigma^j) = W_{p_j(\sigma)}(M_j),
\label{eq:ub51}
  \end{eqnarray} where $C$ is a positive constant independent of $\epsilon$, $\xi_\sigma^j$ is a regular value, $\xi_\sigma = \gamma_j^{-1}\left(\xi_\sigma^j\right)$ satisfies (\ref{eq:regular}) and $W_{p_j(\sigma)}(M_j)$ has been defined in (\ref{eq:ub38}).

\subsection{Explicit representatives}
\label{sec:rep}

In this section, we construct explicit representatives $\nuvec$
for all nonconformal homotopy classes. The general construction procedure is as follows.
Let $0<\eps < \frac{1}{8}$. We partition the domain $Q$ into
four subdomains - (i) $Q_{x,2\eps}$, (ii) $Q_{y,2\eps}$, (iii)
$Q_{z,2\eps}$ and (iv) $Q_0 = Q\setminus \cup_{j} Q_{j,2\eps}$ which
we refer to as the \emph{bulk} domain. Given a nonconformal homotopy class $H$, we specify a set of three non-negative integers $M=\left(M_x, M_y, M_z \right)$ and a conformal or anticonformal homotopy class $H_0 = \left(e_0, k_0, \Omega_0 \right)$ with edge signs $e_0 = \left(e_{0x}, e_{0y}, e_{0z}\right)$ given by 
\begin{equation}
\label{eq:ub54}
e_{0j} = (-1)^{M_j}.
\end{equation} Given $H_0$, there
exists a complex rational representative $F_{H_0}$
of the form (\ref{eq:f}) with this topology(see \cite{mrz2, mrz3}). On each of the
subdomains $Q_{j,2\eps}$, we insert $M_j$ quarter-sphere
configurations. The quarter-sphere configurations are explicitly
given by (\ref{eq:ub22}) and (\ref{eq:ub48}) and we interpolate
between the different quarter-sphere configurations with
negligible energy as in (\ref{eq:ub34}). Given $F_{H_0}$ and the different quarter-sphere configurations on the sub-domains $Q_{j,2\eps}$, we define an overall configuration $K_{H_0,M,\eps}: Q \to \Cc^*$ as shown below -
\begin{equation}\label{eq:ub55}
    K_{H_0,M,\eps}(w) =
\begin{cases}
    F_{H_0}(w), & w  \in Q_0,\\
    G_{j,M_j,\epsilon}(w),& w \in Q_{j,\epsilon}, M_j > 0,\\
   \left((1-s)/G_{j,M_j,\epsilon} + s /F_{H_0}\right)^{-1}(w), &
    w \in Q_{j,2\epsilon}- Q_{j,\epsilon}, M_j > 0, M_j \
    \text{odd},\\
     \left((1-s)G_{j,M_j,\epsilon} + s F_{H_0}\right)(w), &
    w \in Q_{j,2\epsilon}- Q_{j,\epsilon}, M_j > 0, M_j \ \text{even}.
    \end{cases}
\end{equation}
Here $s$ is the switching function on $Q_{2\epsilon} - Q_{\epsilon}$ given by
\begin{equation}\label{eq:ub56}
    s(w) = \frac{\rho - \epsilon}{\epsilon}
\end{equation} and $G_{j,M_j,\epsilon}$ has been defined in (\ref{eq:ub48}). If $M_j=0$ for some $j$, then $K_{H_0,M,\eps} = F_{H_0}$ on $Q_0 \cup Q_{j,2\eps}$. 
We point out that the functions
\newline
$  \left((1-s)/G_{j,M_j,\epsilon} + s /F_{H_0}\right)^{-1}$ and $
\left((1-s)G_{j,M_j,\epsilon} + s
F_{H_0}\right)$ interpolate between
$G_{j,M_j,\epsilon}(w)$ and $F_{H_0}$ on the annular strip
$Q_{j,2\eps} - Q_{j,\eps}$. The representative $\nuvec$ is
then taken to be the inverse projection of $K_{H_0,M,\eps}$ i.e.
$\nuvec = P^{-1}\left(K_{H_0,M,\eps}\right)$.

Let $H = (e, k,\Omega)$ denote an arbitrary nonconformal homotopy class. As discussed in Section~\ref{sec:lb}, we can, without loss of generality, take the edge signs to be
\begin{equation}
\label{eq:ub63}
e_j = +1 \quad \forall j.
\end{equation}
For concreteness, we also assume that the $k_j$'s are ordered as follows - $ 0 < |k_x|\leq |k_y|\leq |k_z|$. As in Section~\ref{sec:lb}, we focus on one representative case $k_j > 0$ for all $j$; the details for the remaining cases are sketched briefly. With $e_j = +1$ and $k_j >0$ for all $j$, the corresponding wrapping numbers are given by (\ref{eq:wrap in terms of e,k,omega}) i.e.
\begin{equation}
\label{eq:ub65}
 w_\sigma = \frac{1}{4\pi}\Omega +
\frac{1}{2}\sum_j \sigma_j k_j + \left(\frac{1}{8} -\delta_{\sigma,+++}\right).
\end{equation} It is easily verified from (\ref{eq:ub65}) that the $w_\sigma$'s are ordered as follows -
\begin{equation}\label{eq:ub66a}
    w_{+++} \ge w_{-++}  \ge w_{+-+} \ge w_{++-}, w_{--+} \ge w_{-+-} \ge w_{+--} \ge w_{---}.
\end{equation} $H$ is nonconformal for $1 \leq w_{+++} \leq  k_x + k_y + k_z - 2$. In this case, $w_{+++} > 0$ is the largest positive wrapping number and $w_{---}<0$ is the smallest negative wrapping number. We consider two different cases according to whether $w_{--+} - w_{++-} = k_z - \left(k_x + k_y \right) < 0$ or $w_{--+} - w_{++-} = k_z - (k_x + k_y) \geq 0$. For convenience, we let $n=w_{+++}$, where $n\in \left[1,k_x + k_y + k_z - 2\right]$ for $H$ nonconformal. Each case above is further divided into
sub-cases according to the value of $n$ and for each
sub-case, we explicitly specify the representative $\nuvec$ in terms of $M=\left(M_x, M_y, M_z\right)$ and a conformal or anticonformal topology $H_0$.

\vspace{.5 cm}

\textbf{Case 1: $k_z - (k_x + k_y) < 0$}

\vspace{.25 cm}

\textbf{Case 1a: $1\leq n \leq k_y - 1$:}

The bulk topology $H_0$:
\begin{eqnarray}
\label{eq:ub66} && e_0 = \left(e_{0x}, e_{0y}, e_{0z}\right)~ \textrm{where} ~
e_{0x} = e_{0y} = e_{0z} = 1 \nonumber\\ && k_0 = \left(k_{0x}, k_{0y}, k_{0z}\right)~ \textrm{where} ~
 k_{0x} = k_x,~ k_{0y} = k_y - n,~ k_{0z} = k_z - n \nonumber \\ && \Omega_0 = -2\pi \sum_j k_j + \frac{7 \pi}{2} + 4n \pi.
\end{eqnarray} The number of quarter-sphere configurations are given by
\begin{equation}
\label{eq:ub67} M_x = 2n,~ M_y = M_z = 0.
\end{equation}

\textbf{Case 1b: $k_y\leq n \leq \frac{\sum_j k_j - 2}{2}$:}

The bulk topology $H_0$:
\begin{eqnarray}
\label{eq:ub68} && e_{0x} = e_{0y} = e_{0z} = 1 \nonumber \\
&& k_{0x} = 1 ,~ k_{0y} = 1 ,~ k_{0z} = \sum_j k_j - 2n - 2 \nonumber\\ && \Omega_0 = -2\pi \sum_j k_j + \frac{7 \pi}{2} + 4n \pi.
\end{eqnarray}
The number of quarter-sphere configurations are given by
\begin{equation} \label{eq:ub68b} M_x = 2\left(n -
k_x + 1\right),~ M_y = 2\left(n - k_y + 1 \right),~ M_z =
2\left(k_x + k_y - n - 2\right).
\end{equation}

\textbf{Case 1c: $\frac{\sum_j k_j - 1}{2}\leq n \leq k_x + k_y -
2$:}

The bulk topology $H_0$:
\begin{eqnarray}
\label{eq:ub69} && e_{0x} = e_{0y} = e_{0z} = 1 \nonumber \\ &&
k_{0x} = 0,~ k_{0y} = 0,~ k_{0z} = 2n + 2 - \sum_j k_j \\ && \Omega_0 = -2\pi\sum_j k_j + \frac{7\pi}{2} + 4n\pi
\end{eqnarray}The number of quarter-sphere configurations are given by
\begin{equation}
\label{eq:69b}
M_x = 2\left(k_y + k_z - n - 1\right),~ M_y =
2\left(k_x + k_z - n - 1\right),~ M_z = 2\left(n - k_z + 1\right).
\end{equation}

\textbf{Case 1d: $k_x + k_y -1 \leq n \leq k_x + k_z - 2$:}

The bulk topology $H_0$:
\begin{eqnarray}
\label{eq:ub70} && e_{0x} = e_{0y} = -1 \quad e_{0z} = 1 \nonumber \\ &&
k_{0x} = 0,~ k_{0y} = 0,~ k_{0z} = 2n + 3 - \sum_j k_j \\ && \Omega_0 = -2\pi\sum_j k_j + \frac{11\pi}{2} + 4n\pi
\end{eqnarray}The number of quarter-sphere configurations are given by
\begin{equation}
\label{eq:ub70b}
M_x = 2\left(k_y + k_z - n - 2\right) + 1,~
 M_y = 2\left(k_x + k_z - n - 2\right) + 1,~
M_z = 2\left(n - k_z + 1\right).
\end{equation}

\textbf{Case 1e: $k_x + k_z-1 \leq n \leq k_y + k_z - 2$:}

The bulk topology $H_0$:
\begin{eqnarray}
\label{eq:ub71} && e_{0x} = e_{0y} = -1 \quad e_{0z} = 1 \nonumber \\&&  k_{0x} = 0,~ k_{0y} = n - k_z + 1,~ k_{0z} = n - k_x - k_y + 2 
\\ && \Omega_0 = -2\pi\sum_j k_j + \frac{11\pi}{2} + 4n\pi.
\end{eqnarray}The number of quarter-sphere configurations are given by
\begin{equation}
\label{eq:ub71b}
M_x = 2\left(k_y + k_z - n - 2\right) + 1,~
 M_y = 2k_x - 1,~ M_z = 0.
\end{equation}

\textbf{Case 1f: $k_y + k_z - 1 \leq n \leq \sum_j k_j - 2$:}

The bulk topology $H_0$:
\begin{eqnarray}
\label{eq:ub72}  && e_{0x} = -1 \quad e_{0y} = e_{0z} = 1 \nonumber \\
&& k_{0x} = k_x,~ k_{0y} = n - k_x - k_z + 1,~ k_{0z} = n - k_x - k_y + 1 
\\ && \Omega_0 = -2\pi\sum_j k_j + \frac{9\pi}{2} + 4n\pi.\end{eqnarray}
The number of quarter-sphere configurations are given by
\begin{equation}\label{eq:ub72b}
 M_x = 2\left(\sum_j k_j - n - 2\right) + 1,~ M_y =  M_z = 0.
\end{equation}

\vspace{.5 cm}

\textbf{Case $2$: $k_z - (k_x + k_y) \geq 0$} \vspace{.25cm}

\textbf{Case $2a$: $1\leq n \leq k_y -1$:} Take $H_0$ and $M=\left(M_x, M_y, M_z\right)$ as in Case $1a$.

\textbf{Case $2b$: $k_y\leq n \leq k_x + k_y -2$:} Take $H_0$ and $M=\left(M_x, M_y, M_z\right)$ as in Case $1b$.

\textbf{Case $2c$: $k_x + k_y -1\leq n \leq k_z - 1$:} This case
is slightly different to the remaining cases discussed in this
section. Firstly, we note that there are precisely four
non-negative wrapping numbers i.e. $w_\sigma \geq 0$ for $\sigma =
(\pm \pm +)$ and four non-positive wrapping numbers i.e. $w_\sigma
\leq 0$ for $\sigma = (\pm \pm -)$ and $\Delta(H) = 0$ for
nonconformal topologies in this range (see (\ref{eq:Delta(H) neq
0})).

As in the preceding cases, we specify the representative $\nuvec$ in terms of a bulk topology $H_0$ and 
three non-negative integers $M=\left(M_x, M_y, M_z\right)$. The bulk topology $H_0$ is anticonformal with invariants
\begin{eqnarray}
\label{eq:ub74} && e_{0x} = 1,~e_{0y} = -1,~ e_{0z} = 1 \nonumber \\ && k_{0x} = k_{0y} = k_{0z} = 0 \nonumber \\ && \Omega_0 = \frac{\pi}{2}.
\end{eqnarray}

We take
\begin{eqnarray}
\label{eq:ub75} && M_x = 2k_y + 2\left( k_z - n - 1\right)
\nonumber \\ && M_y = 2\left(n - k_y \right) + 1
\nonumber \\ && M_z = 0.
\end{eqnarray}
For the sub-domains $Q_{x,\eps}$ and $Q_{y,\eps}$, we define
modified configurations $G^{'}_{x, M_x, \eps}$ and $G^{''}_{y, M_y,
\eps}$ as follows. We consider $G^{'}_{x, M_x, \eps}$ first. Let $1\leq m \leq
M_x$. On the quarter annuli, $2\rho_{m-1} \leq \rho \leq \rho_m$,
we define
\begin{equation}
  \label{eq:ub76}
  g^{'}_{m,\eps}(w) =
  \begin{cases}
    -\frac{w}{\sqrt{\eps} \rho_m}, & 1\leq m \leq 2\left( k_z - n - 1\right),  \text{$m$ odd}, \\
 \frac{\rho_{m-1}}{\sqrt{\eps} w},& 1\leq m \leq 2 \left( k_z - n - 1\right),  \text{$m$ even},  \\
    g_{m,\eps}(w),&  2 \left( k_z - n - 1\right)< m \leq M_x
  \end{cases}
\end{equation}
where $\rho_m$ has been defined in (\ref{eq:ub19}) and $g_{m,\eps}$ in (\ref{eq:ub22}). For $1\leq
m \leq 2\left( k_z - n - 1\right)$ and $m$ odd, $g^{'}_{m,\eps}$
covers the pair of adjacent octants, $\Sigma_{--\pm}$, exactly once with negative orientation whereas for $m$
even, $g^{'}_{m,\eps}$ covers the pair of adjacent octants,
$\Sigma_{+ - \pm}$, exactly once with negative
orientation. For $m > 2 \left( k_z - n - 1\right)$,
$g^{'}_{m,\eps}$ coincides with $g_{m,\eps}$. $\Gamma^{'}_{M_x,\eps}$ and $G^{'}_{x, M_x,
\eps}$ are defined in terms of $g^{'}_{m,\eps}$ by analogy with (\ref{eq:ub37}) and
(\ref{eq:ub48}) respectively.

By analogy with Proposition~\ref{prop:ub1}, we can show that
$G^{'}_{x, M_x, \eps}$ satisfies the tangent boundary conditions
on $Q_{x,\eps}\cap \partial Q$ and the corresponding Dirichlet energy
is bounded from above by
\begin{equation}
\label{eq:ub77} E\left(G^{'}_{x, M_x, \eps}\right) =
\int_{Q_{x,\eps}} \mathcal{H}\left(G^{'}_{x, M_x, \eps}\right)~d^2
w \leq 2\pi M_x + C \eps
\end{equation} where $C$ is a positive constant independent of $\eps$. 
Let
\begin{equation}
  \label{eq:ub78}
 W^{'}_{\sigma} (M_x) =
  \begin{cases}
  - k_y + \left( n - k_z + 1\right) ,& \sigma=\left(\pm - - \right),\\
 \left( n - k_z + 1\right),& \sigma=\left(\pm + - \right),\\
 k_y, & \sigma = (\pm + +), \\
0, & otherwise.
  \end{cases}
\end{equation} Then using arguments similar to
Proposition~\ref{prop:ub1}, one can show that
\begin{equation}
\label{eq:ub79}
 d_{G^{'}_{x,M_x,\epsilon}}(\xi_\sigma^x) =
 W^{'}_{\sigma} (M_x)
\end{equation} where $\xi_\sigma^x$ is a regular value of $G^{'}_{x, M_x, \eps}$ and $\xi_\sigma = \gamma_x^{-1}\left(\xi_\sigma^x\right)$ satisfies (\ref{eq:regular}).

Similarly for $G^{''}_{y, M_y, \eps}$, we define the function
$g^{''}_{m,\eps}$ on the quarter-annuli $2\rho_{m-1} \leq \rho
\leq \rho_m$ for $1\leq m \leq M_y$ as follows -
\begin{equation}
  \label{eq:ub80}
  g^{''}_{m,\eps}(w) =
  \begin{cases}
    \frac{ \wbar}{\sqrt{\eps} \rho_{m}}, & 1\leq m \leq 2\left( n - k_x - k_y +
1\right),  \text{$m$ odd}, \\
 \frac{\rho_{m-1}}{\sqrt{\eps} \wbar},& 1\leq m \leq 2\left( n - k_x - k_y +
1\right),  \text{$m$ even},  \\
    g_{m,\eps}(w),& 2\left( n - k_x - k_y +
1\right)< m \leq M_y.
  \end{cases}
\end{equation}
Then $\Gamma^{''}_{M_y,\eps}$ and $G^{''}_{y, M_y, \eps}$ are
defined by analogy with (\ref{eq:ub37}) and (\ref{eq:ub48})
respectively. For $ 1\leq m \leq 2\left( n - k_x - k_y + 1\right)$
and $m$ odd, $g^{''}_{m,\eps}$ covers the pair of adjacent octants
$\Sigma_\sigma$, $\sigma = ( + - \pm )$ exactly once with positive
orientation and for $m$ even, $g^{''}_{m,\eps}$ covers the pair of
adjacent octants $\Sigma_\sigma$, $\sigma = ( + + \pm )$ exactly
once with positive orientation. For $2\left( n - k_x - k_y +
1\right) < m\leq M_y$, $g^{''}_{m,\eps}$ coincides with $g_{m,\eps}$
defined in (\ref{eq:ub22}). One can directly check that $G^{''}_{y,
M_y, \eps}$ satisfies the tangent boundary conditions on
$Q_{y,\eps}\cap \partial Q$ and has Dirichlet energy 
\begin{equation}
\label{eq:ub81}E\left(G^{''}_{y, M_y, \eps}\right) =
\int_{Q_{y,\eps}} \mathcal{H}\left(G^{''}_{y, M_y, \eps}\right)~d^2
w \leq 2\pi M_y + D \eps
\end{equation} where $D$ is a positive constant independent of $\eps$. The algebraic degrees are readily computed to be 
\begin{equation}
\label{eq:ub82}
 d_{G^{''}_{y,M_y,\epsilon}}(\xi_\sigma^y) =
W^{''}_{\sigma}(M_y)
\end{equation} where $\xi_\sigma^y$ is a regular value of
$G^{''}_{y,M_y,\epsilon}$ and $\xi_\sigma = \gamma_y^{-1}\left(\xi_\sigma^y\right)$ satisfies (\ref{eq:regular}) and 
\begin{equation}
  \label{eq:ub83}
W^{''}_{\sigma}(M_y) =
  \begin{cases}
 k_x - 1 + \left( n - k_x - k_y + 1\right) ,& \sigma=\left(+ \pm + \right),\\
 \left( n - k_x - k_y + 1\right),& \sigma=\left(- \pm + \right),\\ 
 - k_x, & \sigma = (- \pm -), \\
0, & otherwise.
  \end{cases}
\end{equation} 

Given $F_{H_0}$, $G^{'}_{x, M_x, \eps}$ and $G^{''}_{y,M_y,\eps}$,
the function $K_{H_0,M,\eps}$ is defined as in (\ref{eq:ub55}).

\textbf{Case $2d$: $k_z \leq n \leq k_x + k_z-2:$} Take $H_0$ and $M=\left(M_x, M_y, M_z\right)$ as in Case $1d$.

\textbf{Case $2e$: $k_x + k_z-1\leq n \leq k_y + k_z - 2$:} Take $H_0$ and $M=\left(M_x, M_y, M_z\right)$ as in Case $1e$.

\textbf{Case $2f$: $k_y + k_z-1\leq n \leq \sum_j k_j - 2$:} Take $H_0$ and $M=\left(M_x, M_y, M_z\right)$ as in Case $1f$.

\vspace{.5 cm}

\textbf{Remaining cases}

This deals with cases where one or more of the $k_j$'s is either zero or negative. Let $H$ be an arbitrary nonconformal homotopy class with $e_j=+1$ for all $j$ and $k_j \leq 0$ for some $j$. As in cases $1$ and $2$, we can explicitly specify $M=\left(M_x,M_y,M_z\right)$ and $H_0$ in these cases so that the overall representative $\nuvec$ is defined as in (\ref{eq:ub55}). We briefly outline the details here for completeness. We denote the set of wrapping numbers by $\left\{w_\sigma \right\}$. Then the octant $\Sigma_{\sigma_+}$ with $\sigma_+ = \left(\sgn k_x, \sgn k_y, \sgn k_z \right)$ has the largest positive wrapping number and the octant $\Sigma_{\sigma_-}$ with $\sigma_- = \left(-\sgn k_x, -\sgn k_y, -\sgn k_z \right)$ has the smallest negative wrapping number. (There may be more than one octant with largest positive wrapping number or smallest negative wrapping number but we adhere to these choices for definiteness.) 

As before, we look at the triad of octants adjacent to $\Sigma_{\sigma_+}$ and $\Sigma_{\sigma_-}$. We define $M_j$ quarter-sphere configurations on each sub-domain $Q_{j,\eps}$. The conformal quarter-sphere configurations in $Q_{x,\eps}$ cover $\Sigma_{\sigma_-}$ and $\Sigma_{\left(\sgn k_x, -\sgn k_y, -\sgn k_z\right)}$ once with negative orientation (examples of which are the $m$ odd case in (\ref{eq:ub22}))  whereas the anticonformal quarter-sphere configurations in $Q_{x,\eps}$ cover $\Sigma_{\sigma_+}$ and $\Sigma_{(-\sgn k_x, \sgn k_y, \sgn k_z)}$ with positive orientation (examples of which are the $m$ even case in (\ref{eq:ub22})). Similarly, the conformal quarter-sphere configurations in $Q_{y,\eps}$ cover $\Sigma_{\sigma_-}$ and $\Sigma_{\left(-\sgn k_x, \sgn k_y, -\sgn k_z\right)}$ once with negative orientation whereas the anticonformal quarter-sphere configurations in $Q_{y,\eps}$ cover $\Sigma_{\sigma_+}$ and $\Sigma_{(\sgn k_x, -\sgn k_y, \sgn k_z)}$ with positive orientation. Finally, the conformal quarter-sphere configurations in $Q_{z,\eps}$ cover $\Sigma_{\sigma_-}$ and $\Sigma_{\left(-\sgn k_x, -\sgn k_y, \sgn k_z\right)}$ once with negative orientation and the anticonformal quarter-sphere configurations cover $\Sigma_{\sigma_+}$ and $\Sigma_{\left(\sgn k_x, \sgn k_y, -\sgn k_z\right)}$ once with positive orientation.

In each case, the algebraic degrees $d_{G_{x,M_x,\epsilon}}(\sigma), d_{G_{y,M_y,\epsilon}}(\sigma)$ and $d_{\Gamma_{M_z,\epsilon}}(\sigma)$ can be computed as in (\ref{eq:ub51}) and (\ref{eq:ub40}). Once the $M_j$'s are specified, we define the set of numbers $\left\{w_{\sigma,0}\right\}$ as shown below - 
\begin{equation}
\label{eq:case3a}
w_{\sigma,0} = w_\sigma - d_{G_{x,M_x,\epsilon}}(\sigma) -  d_{G_{y,M_y,\epsilon}}(\sigma) - d_{\Gamma_{M_z,\epsilon}}(\sigma).
\end{equation}
The $\left\{w_{\sigma,0}\right\}$'s constitute the set of wrapping numbers for a conformal or anticonformal bulk topology $H_0$. Given $H_0$ and $M = \left(M_x, M_y, M_z\right)$, the representative $\nuvec$ is defined as in (\ref{eq:ub55}).

\begin{lem} \label{lem:nonconf}
 For every nonconformal homotopy class H with wrapping numbers $\left\{w_\sigma\right\}$,
we define the representative $K_{H_0, M,\eps}$ in (\ref{eq:ub55}) where $H_0$ and $M$ are explicitly specified. Let $\left\{w_{\sigma, 0}\right\}$ denote the wrapping numbers of the homotopy class $H_0$. Then
\begin{equation}
\label{eq:nonconf1}
\sum_\sigma |w_{\sigma, 0}| + \left|d_{\Gamma_{M_z,\epsilon}}(\sigma)\right| + \left|d_{G_{x,M_x,\epsilon}}(\sigma)\right| + \left|d_{G_{y,M_y,\epsilon}}(\sigma)\right| = \sum_\sigma |w_\sigma| + \Delta(H)
\end{equation}
where $\Delta(H)$ has been defined in (\ref{eq:Delta(H) neq 0}).
\end{lem}

\begin{proof} For each of the cases in Section~\ref{sec:rep}, we explicitly specify $H_0 = \left(e_0, k_0,\Omega_0\right)$ and a set of three non-negative integers $M=\left(M_x, M_y, M_z\right)$ as shown above. Given $H_0 = \left(e_0, k_0,\Omega_0\right)$, the corresponding wrapping numbers $\left\{w_{\sigma,0}\right\}$ can be computed using formula (\ref{eq:wrap in terms of e,k,omega}). Similarly, given $M_j$, the algebraic degrees $d_{G_{x,M_x,\epsilon}}(\sigma), d_{G_{y,M_y,\epsilon}}(\sigma)$ and $d_{\Gamma_{M_z,\epsilon}}(\sigma)$ are given in (\ref{eq:ub51}) and (\ref{eq:ub40}) respectively where we have dropped explicit reference to regular values, since these algebraic degrees only depend on the octant $\Sigma_\sigma$ in question.  One can directly substitute the values of $w_{\sigma, 0}, d_{G_{x,M_x,\epsilon}}(\sigma), d_{G_{y,M_y,\epsilon}}(\sigma)$ and $d_{\Gamma_{M_z,\epsilon}}(\sigma)$ and check that
\begin{equation}
\label{eq:ub60c}
\sum_\sigma |w_{\sigma, 0}| + \left|d_{\Gamma_{M_z,\epsilon}}(\sigma)\right| + \left|d_{G_{x,M_x,\epsilon}}(\sigma)\right| + \left|d_{G_{y,M_y,\epsilon}}(\sigma)\right| 
= \sum_\sigma |w_\sigma| + \Delta(H)
\end{equation}
in all cases, where $\Delta(H)$ is defined in (\ref{eq:Delta(H) neq 0}). 

We outline the calculations for case $1a$ as an illustration. For case $1a$, the bulk topology $H_0$ is conformal with invariants as in (\ref{eq:ub66}). The corresponding wrapping numbers $\left\{w_{\sigma, 0}\right\}$ are
\begin{eqnarray}
\label{eq:ub66new} && w_{+++,0} = 0, ~ w_{-++,0} = 1- k_x,~ w_{--+,0}
= n - k_x - k_y + 1,~ w_{+-+,0} = n - k_y + 1 \nonumber \\&&
w_{++-,0} = n - k_z +1,~ w_{-+-,0} = n - k_x - k_z + 1,~ w_{---,0}
= 2n - \sum_j k_j + 1,~ w_{+--,0} = 2n - k_y - k_z +
1.
\end{eqnarray} The algebraic degrees $d_{\Gamma_{M_z,\epsilon}}(\sigma)$, $d_{G_{x,M_x,\epsilon}}(\sigma)$ and $d_{G_{y,M_y,\epsilon}}(\sigma)$ are given by (\ref{eq:ub40}) and (\ref{eq:ub51}) respectively i.e.
\begin{equation}
  \label{eq:lem4new}
  d_{G_{x,M_x,\epsilon}}(\sigma) =
  \begin{cases}
    n, & \sigma = \left(\pm + + \right), \\
-n ,& \sigma = \left( \pm - - \right),  \\
    0,& otherwise
  \end{cases}
\end{equation} and $d_{\Gamma_{M_z,\epsilon}}(\sigma) = 0$ and $d_{G_{y,M_y,\epsilon}}(\sigma)=0$ for all $\sigma$ since $M_z = M_y = 0$.

 We substitute these values into the left-hand side of (\ref{eq:nonconf1}) and obtain
\begin{equation}
\label{eq:lem1new}
\sum_\sigma |w_{\sigma, 0}| + \left|d_{\Gamma_{M_z,\epsilon}}(\sigma)\right| + \left|d_{G_{x,M_x,\epsilon}}(\sigma)\right| + \left|d_{G_{y,M_y,\epsilon}}(\sigma)\right| = 
\sum_\sigma |w_\sigma| + 2\min \left(n, k_x -1 \right).
\end{equation}

Next, we compute $\Delta(H)$ for all nonconformal homotopy classes within Case $1a$ i.e. with $1\leq n \leq k_y- 1$. This case can be partitioned into two sub-cases according to whether $n \leq k_x - 1$ or $n\geq k_x - 1$. If $n\leq k_x -1$, then $w_{+++}$ is the only positive wrapping number whereas if $n\geq k_x -1 $, then $w_{+++}$ and $w_{-++}$ are the only two non-negative wrapping numbers. The factor $\chi$ in (\ref{eq:Delta(H) neq 0}) vanishes by definition (see (\ref{eq: chi})) since $k_j > 0$ for all $j$. Here $\sigma_+ = (+++)$ and $\sigma_- = (---)$  in the definition of $\Delta(H)$ in (\ref{eq:Delta(H) neq 0}). We note that
$$ |w_{---}| -     \sum_{\sigma \sim (---)} \Phi(-w_\sigma)  =  \sum_j k_j - n - 1  - \left(2\sum_j k_j  - 3n - 3\right) = 2n + 2 -\sum_j k_j < 0 $$ for $1\leq n \leq k_y - 1$, where $\Phi(x) = \frac{1}{2}\left(x + |x|\right)$. Therefore, we need only compute
$$ |w_{+++}| -     \sum_{\sigma \sim (+++)} \Phi(w_\sigma)   $$ in (\ref{eq:Delta(H) neq 0}).

For $1\leq n \leq k_x - 1$, 
\begin{equation}
\label{eq:nonconfnew2} w_{+++} -  \sum_{\sigma \sim (+++)} \Phi(w_\sigma) = n \end{equation}
whereas for $k_x-1 \leq n \leq k_y - 1$, \begin{equation}
\label{eq:nonconfnew3} w_{+++} -  \sum_{\sigma \sim (+++)} \Phi(w_\sigma) = k_x - 1.\end{equation}
 Combining (\ref{eq:nonconfnew2}) and (\ref{eq:nonconfnew3}), we obtain
\begin{equation}
\label{eq:eq:nonconfnew4}
\Delta(H) = 2\max\left(0,\ w_{+++}  -
    \sum_{\sigma \sim (+++)}
    \Phi(w_\sigma) ,\ |w_{---}| -
    \sum_{\sigma \sim (---)}
    \Phi(-w_\sigma) \right)  = 2\min\left(n, k_x - 1\right)
\end{equation}
and a direct comparison with (\ref{eq:lem1new}) establishes the required result.
\end{proof}

\subsection{Proof of Theorem~\ref{thm:2}}
\label{sec:ub}
\begin{proof}
Let $H$ denote an arbitrary nonconformal homotopy class in $\Ccal_T(O;S^2)$ with associated wrapping numbers $\left\{w_\sigma \right\}$.
For $H$ conformal or anticonformal, $\Delta(H) = 0$ by definition and we can construct a rational representative $F_H$ of the form (\ref{eq:f}). As demonstrated in \cite{mrz2}, $E\left(F_H\right) = \pi \sum_\sigma |w_\sigma|$ (refer to (\ref{eq:ub1})), consistent with the upper bound in Theorem~\ref{thm:2}. 

For $H$ nonconformal, we specify a conformal or anticonformal bulk topology $H_0$ and a set of three non-negative integers $M = \left(M_x, M_y, M_z\right)$. The function $K_{H_0,M,\eps}$ is defined as in (\ref{eq:ub55}) and the representative $\nuvec = P^{-1}\left(K_{H_0,M,\eps}\right)$. 
One can readily verify that
the function $K_{H_0,M,\eps}$ belongs to the space $\Ccal_T(Q,\Cc^*)$. To see why, it suffices to note that $K_{H_0,M,\eps}$
is real on the real axis, imaginary on the imaginary axis and of unit modulus on
the unit circle. This is immediate from the definitions of
$F_{H_0}$ and $G_{j,M_j,\epsilon}$. The interpolatory
functions,
$\left((1-s)/G_{j,M_j,\epsilon} + s/
F_{H_0}\right)^{-1}$ and $
 \left((1-s)G_{j,M_j,\epsilon} + s
F_{H_0}\right)$, on $Q_{j,2\eps} - Q_{j,\eps}$, satisfy the tangent
boundary conditions from the definition of $G_{j,M_j,\epsilon}$ in (\ref{eq:ub48}).  Further, we continuously interpolate between the different quarter-sphere configurations and between $F_{H_0}$ and $G_{j,M_j,\eps}$, and this ensures the continuity of the overall configuration $K_{H_0,M,\eps}$.

Let $\left\{w_{\sigma,K}\right\}$ and $\left\{w_{\sigma,0}\right\}$ denote the wrapping numbers of $K_{H_0,M,\eps}$ and $F_{H_0}$ respectively. Then $w_{\sigma,K}$ and $w_{\sigma,0}$ are related by
\begin{equation}
\label{eq:ub62}
w_{\sigma,K} = w_{\sigma, 0} + d_{G_{x,M_x,\epsilon}}(\sigma) + d_{G_{y,M_y,\epsilon}}(\sigma)+d_{\Gamma_{M_z,\epsilon}}(\sigma)
\end{equation} where $d_{\Gamma_{M_z,\epsilon}}(\sigma)$, $d_{G_{x,M_x,\epsilon}}(\sigma)$ and $d_{G_{y,M_y,\epsilon}}(\sigma)$ are given by (\ref{eq:ub40}) and (\ref{eq:ub51}) respectively (in the special case $2c$, we replace $d_{G_{x,M_x,\epsilon}}(\sigma)$ and $d_{G_{y,M_y,\epsilon}}(\sigma)$ by $d_{G^{'}_{x,M_x,\epsilon}}$ and $d_{G^{''}_{y,M_y,\epsilon}}$ in (\ref{eq:ub79}) and (\ref{eq:ub82}) respectively). One can directly compute the $w_{\sigma,K}$'s from (\ref{eq:ub62}) and check that
\begin{equation}
\label{eq:ub62a}
w_{\sigma,K} = w_\sigma \quad \forall \sigma
\end{equation} so that $\nuvec= P^{-1}\left(K_{H_0,M,\eps}\right)\in H$ as required.

The Dirichlet energy, $E\left(K_{H_0,M,\eps}\right)$, is the sum of the energy contributions from the different sub-domains and
\begin{equation}
\label{eq:ub62b}
E\left(K_{H_0,M,\eps}\right) \leq \int\int_{Q_0} \Hcal(F_{H_0}) d^2 w + \sum_{j} \int\int_{Q_{j,\eps}}\Hcal\left(G_{j,M_j,\eps}\right)~d^2 w + C\eps
\end{equation}
where $C$ is a positive constant independent of $\eps$ (the energy of the interpolatory functions on $Q_{j,2\eps}\setminus Q_{j,\eps}$ has been absorbed into the $C\eps$-contribution on the right-hand side of (\ref{eq:ub62b}) and $G_{z,M_z,\epsilon} = \Gamma_{M_z,\epsilon}$ in (\ref{eq:ub37}) ). The function $K_{H_0,M,\eps}$ is either conformal or anticonformal everywhere by construction. Therefore, by using arguments similar to (\ref{eq:ub1}) and (\ref{eq:ub17}), we have that
\begin{eqnarray}
\label{eq:ub60a}
&& \int\int_{Q_0} \Hcal(F_{H_0}) d^2 w \leq \pi \sum_{\sigma} \left| w_{\sigma, 0} \right|  + C_1 \eps  \nonumber \\
&& \int\int_{Q_{j,\eps}}\Hcal\left(G_{j,M_j,\eps}\right)~d^2 w \leq \pi \sum_{\sigma} \left| d_{G_{j,M_j,\epsilon}}(\sigma) \right| + C_2 \eps, ~j=x,y,z
\end{eqnarray} where $C_1$ and $C_2$ are positive constants independent of $\eps$. We substitute (\ref{eq:ub60a}) into (\ref{eq:ub62b}) to get the upper bound
\begin{equation}
\label{eq:ub60b}
E\left(K_{H_0,M,\eps}\right) \leq \pi \sum_\sigma \left( |w_{\sigma, 0}| + \left|d_{\Gamma_{M_z,\epsilon}}(\sigma)\right| + \left|d_{G_{x,M_x,\epsilon}}(\sigma)\right| + \left|d_{G_{y,M_y,\epsilon}}(\sigma)\right|\right) + D\eps
\end{equation} for a positive constant $D$ independent of $\eps$. Finally, from Lemma~\ref{lem:nonconf}, we have that
$$\sum_\sigma  \left( |w_{\sigma, 0}| + \left|d_{\Gamma_{M_z,\epsilon}}(\sigma)\right| + \left|d_{G_{x,M_x,\epsilon}}(\sigma)\right| + \left|d_{G_{y,M_y,\epsilon}}(\sigma)\right|\right) = \sum_\sigma |w_\sigma| + \Delta(H) $$
and substituting the above into (\ref{eq:ub60b}) yields
\begin{equation}
\label{eq:ub60d}
E\left(K_{H_0,M,\eps}\right) \leq \pi\left(\sum_\sigma |w_\sigma| + \Delta(H) \right) + D\eps.
\end{equation}
In the limit $\eps \to 0$, we recover the upper bound in Theorem~\ref{thm:2}.
\end{proof}

\section*{Acknowledgment}\label{sec:acknowledgment}
AM was supported by a Royal Commission for the Exhibition of 1851
Research Fellowship between 2006 - 2008. AM is now supported by
Award No. KUK-C1-013-04 , made by King Abdullah University of
Science and Technology (KAUST) to the Oxford Centre for
Collaborative Applied Mathematics. We thank Tim Riley and Ulrike Tillmann for
stimulating discussions.

\appendix

 \renewcommand{\theequation}{A.\arabic{equation}}
  \setcounter{equation}{0}  

\section{Proof of Propositions \ref{prop: spelling length} and \ref{prop: spelling length 2}}



\subsection{Spelling length on words.} \label{sec: free group}

We state some basic definitions and notation concerning words and
free groups. Let
\begin{equation}
  \label{eq:alphabet}
  {\cal A}_N = \{ X_1, \ldots, X_N, X_1^{-1}, \ldots, X_N^{-1}\}.
\end{equation}
Elements of ${\cal A}_N$ are called {\it letters}.
$X_r^{-1}$ is called the {\it
  inverse} of $X_r$, and vice versa.
Sometimes we denote letters  by $A$, $B$, $C$, etc.
A {\it word of length $k$}
on ${\cal A}_N$ is a $k$-tuple of letters ${\bf U} =
(U(1),\ldots,U(k))$, where $U(j) \in {\cal A}_N$. The set of words
of length $k$ is denoted ${\cal
  L}_N^k$.  The word of zero length is denoted $\bf e$, and we write
  ${\cal L}_N^0 = \{ {\bf e} \}$.  The set of all words is given by
\begin{equation}
    \label{eq:L_N}
    {\cal L}_N = \cup_{k=0}^\infty {\cal L}_N^k.
\end{equation}
Let $L$ denote the length function on ${\cal L}_N$, so that $L({\bf
  U}) = k$ for $\Ub \in {\cal L}_N^k$.  
Adjunction defines a product operation on ${\cal L}_N$; given ${\bf
U} \in {\cal L}_N^k$ and ${\bf V} \in {\cal
  L}_N^l$, define $({\bf U}, {\bf V}) \in  {\cal L}_N^{k+l}$ by
\begin{equation}
  \label{eq:word_mult}
  ({\bf U}, {\bf V}) = (U(1),\ldots, U(k), V(1), \ldots, V(l)).
\end{equation}
Then $({\bf
  U},{\bf e}) = ({\bf e},{\bf U}) = {\bf U}$.  Clearly,
\begin{equation}
  \label{eq:length_is_additive}
  L(({\bf U},{\bf V})) = L({\bf U}) + L({\bf V}).
\end{equation}
We use the exponential notation $X_r^j$ to denote the $j$-tuple
$(X_r,\ldots, X_r)$ for $j$ positive; for $j$ negative, $X_r^j$
denotes the $j$-tuple $(X_r^{-1},\ldots, X_r^{-1})$, and for $j =
0$, the identity $\bf e$. Also, for ${\bf U} \in {\cal L}_N^k$ with $k>0$,
define ${\bf U^{-1}} \in {\cal L}_N^k$ by
\begin{equation}
  \label{eq:inverse_words}
  {\bf U^{-1}} = (U(k)^{-1},\ldots,U(1)^{-1})
\end{equation} and define $\bf e^{-1}$ to be $\bf e$.

The {\it free group} $F(X_1,\ldots,X_N)$ is the set of equivalence
classes in ${\cal
  L}_N$ under all relations of the form
\begin{equation}
  \label{eq:free_relation}
  ({\bf U}, X_r, X_r^{-1}, {\bf V}) \sim ({\bf U}, X_r^{-1}, X_r, {\bf V}) \sim
  ({\bf U},{\bf V}).
\end{equation}
Given ${\bf U} \in {\cal L}_N$, we denote its equivalence class in
$F(X_1,\ldots,X_N)$ either by $[{\bf U}]$ or by  $U$. Multiplication
in $F(X_1,\ldots,X_N)$ is defined and denoted by
\begin{equation}
  \label{eq:mult_in_FN}
  UV = [({\bf U}, {\bf V})]
\end{equation}
Inverses in $F(X_1,\ldots,X_N)$ are given by
\begin{equation}
  \label{eq:inverses_1}
  U^{-1} = [{\bf U^{-1}}].
\end{equation}


We introduce another length function on ${\cal L}_N$, denoted
$\lambda$. As it will turn out to be equivalent to the spelling
length $\Lambda$ (cf Proposition~\ref{prop:lambda=Lambda} below), we
shall also refer to $\lambda$ as the spelling length. $\lambda$ is
defined inductively as follows: On words of length $0$, ie $\bf e$,
we take
\begin{equation}
  \label{eq:lambda(e)}
  \lambda({\bf e}) = 0.
\end{equation}
Given that $\lambda$ is defined on words of length less than $k$,
for ${\bf U} \in {\cal L}_n^k$ we define
\begin{equation}
  \label{eq:formula_for_l}
  \lambda({\bf U}) = \min\Big(
1 + \lambda(\Ub_{2:k}),
\min_{  U(j) = U(1)^{-1}} \lambda(\Ub_{2:j-1}) + \lambda(\Ub_{j+1:k}
) \Big),
\end{equation}
where ${\bf U}_{a:b}$ is equal to $(U(a), \ldots, U(b))$ for $b \ge
a$ and to $\bf e$ for $b < a$.
From the definition (\ref{eq:formula_for_l}), if $\lambda({\bf U}) <
1 + \lambda(\Ub_{2:k})$, then there is an index $j$ such that $U(1)$
and $U(j)$ are inverses, and $ \lambda({\bf U}) =
\lambda(\Ub_{2:j-1}) + \lambda(\Ub_{j+1:k})$. Proceeding recursively,
we see that $\lambda({\bf U})$ is achieved by specifying a set of
inverse-letter pairs
\begin{equation}
  \label{eq:pairing}
  I = \left \{\{a_1,b_1\}, \ldots, \{a_p, b_p\}\right\}
\end{equation}
such that
\begin{subequations} \label{eq:pairing_properties_2}
\begin{gather}
  U(b_i) = U(a_i)^{-1} \ \text{for all}\ \{a_i,b_i\} \in I,\\
  a_i < a_j  \ \text{implies  either} \ b_i < a_j \ \text{or} \
\ b_i > b_j.
\end{gather}
\end{subequations}
The last condition just means that the intervals $[a_i,b_i]$ and
$[a_j, b_j]$ are either disjoint or else one contains  the other as
a proper subset. $I$ is called a {\it pairing}.  Given a word $\bf
U$ and a pairing $I$, $I$ is said to be {\it valid} for $\bf U$ if
(\ref{eq:pairing_properties_2}) is satisfied. If $i$ belongs to some
pair in $I$, we say that $i$ is {\it paired in $I$}.  Paired letters
do not contribute to the spelling length.  Hence, if $I$ is a valid
pairing for $\bf U$,
\begin{equation}
  \label{eq:pairing_and_lambda}
  \lambda({\bf U}) \le L({\bf U}) - 2 |I|,
\end{equation}
where $|I|$ is the number of elements (ie, pairs) in $I$. If
$\lambda({ \bf U}) = L({\bf U}) - 2|I|$, we say that $I$ is an {\it
optimal pairing}. The preceding discussion implies that every word
has an optimal pairing.

%
%

$\lambda$ and $\Lambda$ are equivalent in the following sense:
\begin{prop}\label{prop:lambda=Lambda}
 $\lambda$ descends to a function on
the free group $F(X_1,\ldots,X_N)$ where it coincides with
$\Lambda$.  That is, if $[{\bf U}] = [{\bf U'}] = U$, then
\[ \lambda({\bf U}) = \lambda({\bf U'}) = \Lambda(U).\]
%
\end{prop}
The proof of Proposition~\ref{prop:lambda=Lambda} makes use of
several properties of $\lambda$ which are established in the
following lemmas.
\begin{lem}[Sub-additivity]\label{lem:subadd}
For all  $\Ub, \Vb \in \Lcal_N$,
\[\lambda(({\bf U},{\bf V})) \le \lambda({\bf U}) + \lambda({\bf V}).\]
\end{lem}

\begin{proof}
By induction on $L({\bf U})$.  The statement is trivial for $L({\bf
U}) = 0$, ie ${\bf U} = {\bf e}$.  Suppose it is true for all ${\bf
U}$ with $L({\bf U}) < k$, and suppose $L({\bf U}) = k$. From the
definition (\ref{eq:formula_for_l}),
\begin{equation}
  \label{eq:prop_1_1}
   \lambda((\Ub,\Vb)) \le
 \min\Big(
1 + \lambda(\Ub_{2:k}, \Vb), \min_{ U(j) = U(1)^{-1}}
\lambda(\Ub_{2:j-1}) + \lambda(\Ub_{j+1:k}, \Vb ) \Big)
\end{equation}
(we have omitted terms from indices $j$ for which
 $V(j) = U(1)^{-1}$ -- hence we have an inequality rather than an equality in
(\ref{eq:prop_1_1})). By the induction hypothesis,
\begin{equation}
  \label{eq:prop_1_2}
\lambda( \Ub_{2:k}, \Vb) \le \lambda(\Ub_{2:k}) + \lambda(\Vb),
\quad \lambda(\Ub_{j+1:k}, \Vb )\le \lambda(\Ub_{j+1:k}) +
\lambda({\bf V}).
\end{equation}
From (\ref{eq:formula_for_l}), (\ref{eq:prop_1_1}) and
(\ref{eq:prop_1_2}),
  \begin{equation}
  \label{eq:prop_1_3}
   \lambda(({\bf U},{\bf V})) \le
 \min\Big(
1 + \lambda(\Ub_{2:k}),  \min_{ U(j) = U(1)^{-1}}
\lambda(\Ub_{2:j-1}) +
 \lambda(\Ub_{j+1:k})
\Big) + \lambda({\bf V}) =
 \lambda({\bf U})+ \lambda({\bf V}).
\end{equation}
\end{proof}

\begin{lem}[Cyclicity]\label{lem:cyclic}
Let $\Ub = (U(1), \ldots, U(k))$.  Then
\[
\lambda((U(k),U(1)\ldots U(k-1))) =  \lambda({\bf U}).\]
\end{lem}

\begin{proof}
By induction on $L({\bf U})$.  The statement is trivial for $L({\bf
U}) = 0$ and $L({\bf U}) = 1$.  Given $k > 1$, suppose it is true
for all ${\bf U}$ with $L({\bf U}) < k$ and let $L({\bf U}) =
k$. Let $\Ubp = (U(k),U(1),\ldots,U(k-1)$.

Let us compute $\lambda({\bf U)}$, applying the definition
(\ref{eq:formula_for_l}) twice, as follows: The first application
gives
\begin{equation}
  \label{eq:p2_1}
  \lambda({\bf U}) =
\min\left( 1 + \lambda(\Ub_{2:k}), \min_{U(i) = U(1)^{-1}}
\lambda(\Ub_{2:i-1}) +
 \lambda(\Ub_{i+1:k})
\right).
\end{equation}
The induction hypothesis implies that
\begin{equation}
  \label{eq:p2_2}
    \lambda(\Ub_{2:k}) = \lambda((U(k),\Ub_{2:k-1})), \quad
  \lambda(\Ub_{i+1:k}) = \lambda((U(k),\Ub_{i+1:k-1})).
\end{equation}
Substituting 
into (\ref{eq:p2_1}) and applying
(\ref{eq:formula_for_l}) again to terms in which $U(k)$ appears as
the first letter, we get that
\begin{multline}
  \label{eq:p2_3}
  \lambda({\bf U}) =
\min\Big( 2 + \lambda(\Ub_{2:k-1}),
1 + \min_{ U(j) = U(k)^{-1}} \lambda(\Ub_{2:j-1}) +
\lambda(\Ub_{i+1:k-1}),\\ 1 + \min_{ U(i) = U(1)^{-1}} \lambda(\Ub_{2:i-1}) +
\lambda(\Ub_{j+1:k-1}),\\
 \min_{ U(i) = U(1)^{-1}}\
\min_{ U(j) = U(k)^{-1}\atop j >i + 1} \lambda(\Ub_{2:i-1}) +
\lambda(\Ub_{i+1:j-1})  + \lambda(\Ub_{j+1:k-1})\Big).
\end{multline}

Next we compute $\lambda(\Ubp)$, applying the definition
(\ref{eq:formula_for_l}) twice, as follows:  The first application
gives
\begin{equation}
  \label{eq:p2_1b}
  \lambda(\Ubp) =
\min\left( 1 + \lambda(\Ub_{1:k-1}),\min_{U(i) = U(k)^{-1}}
\lambda(\Ub_{1:i-1})  +
 \lambda(\Ub_{i+1:k-1})
\right).
\end{equation}
Applying (\ref{eq:formula_for_l})  to terms in which the first
letter of the argument of $\lambda$ is $U(1)$, we get
\begin{multline}
  \label{eq:p2_3b}
  \lambda({\bf \Ubp}) =
\min\Big(
2 + \lambda(\Ub_{2:k-1}),
1 + \min_{ U(j) = U(1)^{-1}} \lambda(\Ub_{2:j-1})  +
\lambda(\Ub_{j+1:k-1}), \\ 1 + \min_{ U(i) = U(k)^{-1}}
\lambda(\Ub_{2:i-1})  + \lambda(\Ub_{i+1:k-1}),\\
 \min_{ U(i) = U(k)^{-1}}\
\min_{U(j) = U(1)^{-1}\atop j < i - 1} \lambda(\Ub_{2:j-1})
 + \lambda(\Ub_{j+1:i-1}) +
\lambda(\Ub_{i+1:k-1})\Big).
\end{multline}
Comparison of (\ref{eq:p2_3}) and (\ref{eq:p2_3b}) shows that
$\lambda({\bf U}) = \lambda({\bf U'})$.
\end{proof}

\begin{lem}[Zero length words]\label{lem:zerolen}
\[ \lambda({\bf U}) = 0 \ \text{if and only if} \ U = e.\]
\end{lem}
\begin{proof}
  First, we suppose that $\lambda({\bf U}) = 0$.
  We proceed by induction on $L({\bf U})$. The assertion
  is true for $L({\bf U})= 0$, by definition.  Suppose it is true for all ${\bf U}$
  of length less than $k$, and let $\bf U$ be a word of length $k$
  with $\lambda({\bf U}) = 0$.  Then there is some $j$ with $1 < j \le
  k$ such that $U(j) = U(1)^{-1}$ and
\begin{equation}\label{eq:pzero_1}
  0 = \lambda({\bf U}) = \lambda(\Ub_{2:j-1}) + \lambda(\Ub_{j+1:k}).
\end{equation}
Since $\lambda$ is nonnegative, it follows that
$\lambda(\Ub_{2:j-1})= \lambda(\Ub_{j+1:k}) = 0$. By the induction
hypothesis, it follows that $[\Ub_{2:j-1}] =  [\Ub_{j+1:k}]= e$.
Therefore,
\begin{equation}
  \label{eq:pzero_2}
  U= U(1) \, [\Ub_{2:j-1}] \, U(j)\, [\Ub_{j+1:k}] = U(1)
  U(j)= e.
\end{equation}

Next,  suppose that $ U = e$. Then $L({\bf U})$ is even and
\begin{equation}
  \label{eq:pzero_3}
  {\bf U} = (X_{r_1},\ldots, X_{r_m},X_{r_m}^{-1}, \ldots, X_{r_1}^{-1}).
\end{equation}
It follows that
$  I = \{\{1,2m\}, \{2,2m-1\},\ldots, \{m,m+1\}\}$
is a valid pairing for $\bf U$ and that $L({\bf U}) - 2|I| = 2m - 2m =
0$.  From (\ref{eq:pairing_and_lambda}) and the fact that $\lambda$
is nonnegative, it follows that
 $ \lambda({\bf U}) = 0$.
\end{proof}

\begin{lem}\label{lem:invconj}
If $\hb \in {\cal L}_N$ and $X \in {\cal A}_N$, then
\[ \lambda(({\hb},X,{\hbi})) = 1.\]
\end{lem}

\begin{proof}
Let $L(\hb) = k$.  Then
$  I = \{\{1,2k+1\}, \{2,2k\},\ldots, \{k,k+2\}\}$
is a valid pairing for $(\hb, X, \hbi)$, so that, from
(\ref{eq:pairing_and_lambda}),
\begin{equation}
  \label{eq:pinvconj_2}
   \lambda\left(\hb, X, \hbi\right) \le L\left(\hb, X,
     \hbi\right) - 2|I | = 2k+1 - 2k = 1.
\end{equation}
On the other hand, since $[(\hb, X, \hbi)] \ne e$, it follows from
Lemma~\ref{lem:zerolen} that $\lambda(\hb, X, \hbi) > 0$. Therefore,
we may conclude that
$   \lambda\left(\hb, X, \hbi\right) = 1$.

\end{proof}
We proceed to  the proof of Proposition~\ref{prop:lambda=Lambda}.
\begin{proof}[Proof of Proposition~\ref{prop:lambda=Lambda}]
First, we show that $\lambda(\Ub) = \lambda(\Ubp)$ for $[\Ub] =
[\Ubp]$.  In view of the defining relations (\ref{eq:free_relation})
for $F(X_1,\ldots,X_N)$, it suffices to show that
\begin{equation}
  \label{eq:pdef_0}
  \lambda(({\bf U}, X^{-1}, X, {\bf
  V})) =  \lambda
(({\bf U}, X, X^{-1}, {\bf V})) =  \lambda (({\bf U}, {\bf V})).
\end{equation}
We will just consider $\lambda(({\bf U}, X^{-1}, X, {\bf
  V}))$; the argument for $\lambda
(({\bf U}, X, X^{-1}, {\bf V}))$ is similar.

By cyclicity (Lemma~\ref{lem:cyclic}), it suffices to show that for $\bf W\in {\cal L}_N^k$
\begin{equation}
  \label{eq:pdef_1}
  \lambda
(( X^{-1}, X, {\bf W})) =
  \lambda
({\bf W}).
\end{equation}
From (\ref{eq:formula_for_l}),
\begin{equation}
  \label{eq:pdef_2}
  \lambda
(( X^{-1}, X, {\bf W})) = \min\Big( 1 + \lambda((X,{\bf W})),
\lambda({\bf W}), \min_{W(j) = X} \lambda((X, \Wb_{1:j-1})) +
\lambda(\Wb_{j+1:k})\Big).
\end{equation}
To establish (\ref{eq:pdef_1}), 
we show that the first and third members of the right-hand side
of (\ref{eq:pdef_2}) are not smaller than the second, namely $\lambda({\bf
  W})$.  We start with the first member, namely  $ 1 + \lambda((X,{\bf W}))$.
Applying (\ref{eq:formula_for_l}), 
we get that
\begin{equation}
  \label{eq:pdef_3}
  1 + \lambda((X,{\bf W})) =
\min\Big( 2 + \lambda({\bf W}),
 \min_{W(j) = X^{-1}}1 +  \lambda(\Wb_{1:j-1})  + \lambda(\Wb_{j+1:k}) \Big).
\end{equation}
By subadditivity (Lemma~\ref{lem:subadd}),
\begin{equation}
  \label{eq:pdef_4}
  1 + \lambda(\Wb_{1:j-1})  + \lambda(\Wb_{j+1:k})
  \ge \lambda({\bf W}).
\end{equation}
 Therefore, from (\ref{eq:pdef_3}) and (\ref{eq:pdef_4}),
\begin{equation}
  \label{eq:pdef_5}
   1 + \lambda((X,{\bf W})) \ge \lambda({\bf W}).
\end{equation}
Referring to the third member of the right-hand side of
(\ref{eq:pdef_2}), we have, for $W(j) = X$, that $\lambda((X,
\Wb_{1:j-1})) = \lambda(\Wb_{1:j})$ (cyclicity again), so that, by
subadditivity,
\begin{equation}  \label{eq:pdef_6}
 \lambda((X, \Wb_{1:j-1})) +
\lambda(\Wb_{j+1:k})=
\lambda(\Wb_{1:j}) + \lambda(\Wb_{j+1:k}) \ge \lambda({\bf W}),
\end{equation}
as required.

Next, we show that $\lambda(\Ub) = \Lambda(U)$.   We proceed by
induction. For words of length zero, this follows from
Lemma~\ref{lem:zerolen}. Suppose the statement is true for words of
length less than $k$, and let $\bf U$ have length $k$. Let $n =
\Lambda(U)$.  Then $U$ has a spelling of length $n$, ie
\begin{equation}
  \label{eq:plam_1}
  U = h_1 X_{r_1} h_1^{-1} \cdots   h_n X_{r_n} h_n^{-1}
\end{equation}
for some $h_i \in F(X_1,\ldots,X_N)$ and  $X_{r_i} \in {\cal
  A}_N$.  Let ${\hbsi}$ be words corresponding to $h_i$.  Then
  \begin{equation}
    \label{eq:plam_2}
    \lambda({\bf U}) = \lambda(({\bf h_1}, X_{r_1}, {\bf h_1^{-1}},
 \ldots,
{\bf  h_n},   X_{r_n}, {\bf h_n^{-1}})).
  \end{equation}
  From subadditivity (Lemma~\ref{lem:subadd}) and
Lemma~\ref{lem:invconj}, it follows that
\begin{equation}
  \label{eq:plam_3}
   \lambda({\bf U}) \le \lambda(({\bf h_1}, X_{r_1}, {\bf h_1^{-1}})) +
   \cdots + \lambda(({\bf  h_n},   X_{r_n}, {\bf h_n^{-1}})) = n = \Lambda(U).
\end{equation}

It remains to show that
   $\lambda({\bf U}) \ge \Lambda(U)$.
From (\ref{eq:formula_for_l}), we have either that
\begin{equation}
  \label{eq:plam_5}
   \lambda({\bf U}) = 1 + \lambda(\Ub_{2:k})
\end{equation}
or that
\begin{equation}
  \label{eq:plam_6}
   \lambda({\bf U}) = \lambda(\Ub_{2:j-1}) + \lambda(\Ub_{j+1:k})
\end{equation}
for some $j$ with $U(j) = U(1)^{-1}$.  In case (\ref{eq:plam_5})
holds, use the induction hypothesis to conclude that
\begin{equation}
  \label{eq:plam_7}
   \lambda({\bf U}) = 1 + \Lambda([{\bf U}_{2:k}]) \ge \Lambda\left(U(1)[{\bf U}_{2:k}]\right)
   = \Lambda(U),
\end{equation}
where we have used the fact (easily verified) that $\Lambda$ is
subadditive, ie $\Lambda(UV) \le \Lambda(U) + \Lambda(V)$.  On the
other hand, if (\ref{eq:plam_6}) holds, then the induction
hypothesis implies that
\begin{multline}
  \label{eq:plam_8}
   \lambda({\bf U}) =
 \Lambda([\Ub_{2:j-1}]) + \Lambda([\Ub_{j+1:k}])
= \Lambda\left(U(1)\, [\Ub_{2:j-1}]  U(1)^{-1})\right) +
 \Lambda([\Ub_{j+1:k}])\\
= \Lambda([\Ub_{1:j}]) + \Lambda([\Ub_{j+1:k}]) \ge \Lambda(U),
\end{multline}
where in the second equation we have used the invariance of
$\Lambda$ under conjugation (easily verified) and in the third the
subadditivity of $\Lambda$.
\end{proof}

\subsection{Set products of conjugacy classes of words}
\label{sec: spell length and class product}

Let
\begin{equation}
  \label{eq:What}
  \What = \{U_1, \ldots, U_q \, | \, U_j \in F(X_1,\ldots,X_N)\}
\end{equation}
denote a set of elements of the free group, and let
\begin{equation}
   \label{eq:CP_group}
   \Vcal(\What) = \langle  U_1\rangle \cdots \langle  U_q\rangle
\end{equation}
denote the set product of their conjugacy classes. We wish to
determine the minimum of the spelling length over
$\Vcal(\What)$.
In view of Proposition~\ref{prop:lambda=Lambda}, we can work with
words rather than elements of $F(X_1,\ldots,X_N)$.  Choose words
${\bf U_j}$ so that $[{\bf U_j}] = U_j$, and let
\begin{equation}
  \label{eq:CP}
 \Vcalb({\Whatb})  = \{
({\bf h_1}, {\bf U_1}, {\bf h_1^{-1}}, \ldots, {\bf h_q}, {\bf U_q},
{\bf h_q^{-1}}) \, | \, {\bf h_j} \in {\cal L}_N\}.
\end{equation}
Thus,  every $U \in \Vcal({\What})$ has a representative $\Ub \in
\Vcalb({\Whatb})$ with $[\Ub] = U$, and
\begin{equation}
  \label{eq:lambda_form}
  \min_{U \in \Vcal(\What)} \Lambda(U)  =  \min_{{\bf U} \in  \Vcalb({\Whatb})}
  \lambda({\bf U}).
\end{equation}

Let
\begin{equation}
  \label{eq:U=}
  {\bf U}=
({\bf h_1}, {\bf U_1}, {\bf h_1^{-1}}, \ldots, {\bf h_q}, {\bf U_q},
{\bf h_q^{-1}}) \in \Vcalb({\Whatb}),
\end{equation}
 and suppose that $L({\bf U}) = k$.
Let $\Ical(\Ub) = \{1,\ldots, k\}$ denote the indices of the letters
in $\Ub$.
Let ${\cal
  W}({\bf U})$ denote the set of indices of the letters in the
$\bf U_j$'s, and ${\cal C}({\bf U})$ the set of indices of the
letters
  in the
$\bf h_j$'s and $\bf h_j^{-1}$'s, so that
\begin{equation}
  \label{eq:I_decomp}
  \Wcal(\Ub) \cup \Ccal(\Ub)
  = \Ical(\Ub).
\end{equation}
Indices in ${\cal C}({\bf U})$
naturally fall into pairs associated to conjugate letters in $\bf
h_j$ and $\bf h_j^{-1}$. For example, if ${\bf h_j} = (A,B^{-1},C)$,
then ${\bf
  h_j^{-1}} = (C^{-1}, B, A^{-1})$, and we say that the
indices of $A$ and $A^{-1}$ are conjugate, as are the indices of $B$
and $B^{-1}$  and of $C$ and $C^{-1}$.  In general, for $c \in {\cal
C}({\bf U})$,
let $\bar c \in {\cal C}({\bf U})$ denote the index to which it is
 conjugate. One can compute an explicit formula for $\bar c$ but we won't be needing explicit formulae for this discussion.

Let $I$ be a valid pairing for $\Ub$.  We say that $i, j \in
\Wcal({\Ub})$ are {\it linked in $I$} if there exists a sequence of
indices
$c_1, \ldots, c_m$ in $\Ccal(\Ub)$ such that
\begin{equation}
  \label{eq:linking_sequence}
  \{i,c_1\}, \{\bar c_1,
c_2\},\ldots, \{\bar c_{m-1},c_m\}, \{c_m,j\} \in I.
\end{equation}
We call $c_1, \ldots, c_m$ a {\it linking sequence}.  If $i$ and $j$
are linked, the linking sequence between them is unique.
It is clear that, if $i$ and $j$ are linked, then $U(j) =
U(i)^{-1}$. 
Moreover, the $U(c_r)$'s are all equal to $U(i)^{-1}$, while the
$U(\bar c_r)$'s are all equal to $U(i)$. We say that a pairing $I$
is {\it reduced}  if every $i \in {\cal W}({\bf U})$
is either unpaired or else is linked to some $j \in {\cal W}({\bf
U})$.

We next describe a procedure for removing letters from a word.  Let
$\Ub \in \Lcal_N$ be a word of length $k$, and let $R$ be a subset
of $\Ical(\Ub)$.  Define the {\it re-indexing map} $\psi_R$ to be
the bijection between $\{1,\ldots, k\} - R$ and $\{1,\ldots, k -
|R|\}$ given by
\begin{equation}
  \label{eq:psi_R}
  \psi_R(i) = i - \left| \{ j \in R \, | \, j < i \} \right|
\end{equation}
(ie, $\psi_R(i)$ is $i$ minus the number of indices in $R$ less than
$i$).  Define $\Vb \in \Lcal_N$ to be the word of length $k - |R|$
given by
\begin{equation}
  \label{eq:V_from_U}
  \Vb = (U(\psi^{-1}_R(1)),\ldots,U(\psi^{-1}_R(k - |R|)))
\end{equation}
(so $\Vb$ is just $\Ub$ without the letters indexed by $R$).  We say
that $\Vb$ is the word obtained by {\it
  removing $R$ from $\Ub$}.  Note that we may regard $\psi_R$ as a
bijection between $\Ical(\Ub) - R$ and $\Ical(\Vb)$. If $\Ub$
belongs to $\Vcalb(\Whatb)$, then, in general, $\Vb$ does not.
However, if $R\subset \Ccal(\Ub)$ and $c \in R$ implies that $\bar c
\in R$, then removing $R$ amounts to replacing one or more the $\bf
h_j$'s by shorter words, and $\Vb$ belongs to $\Vcalb(\Whatb)$ as
well.

Given $I$, a valid pairing for $\Ub$, let us define a pairing $J$
by
\begin{equation}
  \label{eq:J_construct}
  J = \{\{\psi_R(i),\psi_R(j\}\} \,| \, i,j \in \Ical(\Ub) - R\ \text{and}\
  \{i,j\} \in I\}.
\end{equation}
That is, $J$ contains all the re-indexed pairs of indices in $I$
which haven't been removed from $\Ub$.  It is straightforward to
check that $J$ is valid for $\Vb$ (cf
(\ref{eq:pairing_properties_2})).  However, $I$ optimal does not
imply that $J$ is optimal, nor does $I$ reduced imply that $J$ is
reduced. We say that $J$ is the pairing obtained from {\it removing
$R$ from $I$}.

The following proposition shows that the minimum spelling length on
$\Vcal(\What)$ can be realised by a word with an optimal reduced
pairing.
\begin{prop}\label{prop:reduced_exist} Let
$   m = \min_{U \in \Vcal(\What)} \Lambda(U)$.
Then there exists  ${\bf U} \in \Vcalb({\Whatb})$ with  optimal reduced
pairing $I$ such that
$\lambda({\bf U}) = m$.
\end{prop}
\begin{proof}
Choose ${\bf U'} \in \Vcalb({\Whatb})$ such that $\lambda({\bf
U'}) = m$.
Let $I'$ be an optimal pairing for $\bf U'$ (as discussed in
Section~\ref{sec: free group}, such an optimal pairing exists).  Let $t$ denote the number of indices in $\Wcal({\bf
U'})$ which are paired in $I'$ but which are not linked to an
index in $\Wcal({\bf U'})$. If $t=0$, then $I'$ is reduced, and we
are done.  In what follows, we obtain a word $\bf U \in \Vcalb({\Whatb})$ with $\lambda({\bf U}) = m$ that has $(t-1)$ paired but unlinked indices in $\Wcal({\bf U})$.  Let $i \in
\Wcal({\bf U'})$ be an index which is paired in $I'$ but which is
not linked to an index in $\Wcal({\bf U'})$. Then $i$ is paired
with some (unique) $c_1 \in \Ccal({\bf U'})$. Either $\bar c_1$ is
unpaired, or else $\bar c_1$ is paired to some (unique) $c_2 \in
\Ccal({\bf U'})$. Continuing in this way, we produce a sequence of
indices $c_1, \ldots, c_d \in \Ccal$, where $\{i,c_1\}$, $\{\bar
c_1,c_2\}, \ldots, \{\bar c_{d-1},c_d\}$ are paired in $I'$ and $\bar c_d$
is unpaired.

We remove the set of indices $R = \{c_1, \bar c_1, \ldots, c_d, \bar
c_d\}$ from $\bf U'$ and $I'$ to obtain $\bf U$ with valid
pairing $I$.  Since $R \subset \Ccal({\bf U'})$ and the elements
of $R$ occur in conjugate pairs, we have that $\bf U \in
\Vcalb(\Whatb)$. It is clear that $L({\bf U}) = L({\bf U'}) - 2d$
and $|I| = |I'| - d$.
From (\ref{eq:pairing_and_lambda}), it follows that
\begin{equation}\label{eq:U and U'}
  \lambda({\bf U}) \le L({\bf U}) - 2| I| =  L({\bf U'}) -
2|I'| =
  \lambda({\bf U'}) = m,
\end{equation}
where the second-to-last equality holds because $I'$ is, by assumption,
optimal.  On the other hand, Proposition~\ref{prop:lambda=Lambda}
implies that $ \lambda({\bf U}) \ge m$.  Therefore, $ \lambda({\bf U})  = m$.
By construction, $\bf U$ has $t-1$ indices
in $\Wcal({\bf U})$ which are paired in $I$ but not linked to
indices in $\Wcal({\bf U})$ (in particular, while $i \in \Wcal({\bf
  U'})$ is such an index, $\psi_R(i) \in \Wcal({\bf U})$ is, by
construction, unpaired, and therefore is not). One can repeat the construction $t$ times to get a word in $\Whatb$, with optimal reduced pairing, that achieves the minimum spelling length.
\end{proof}

Let $\bf U \in \Vcalb(\Whatb)$ and let $I$ be an optimal reduced
pairing for $\bf U$.  Let $[ I ]_W$ denote the number of indices in
$\Wcal(\Ub)$ which are paired in $I$, and $[ I ]_C$ the number of
indices in $\Ccal(\Ub)$ which are paired in $I$, so that
\begin{equation}
  \label{eq:<I>}
  2| I | = [I ]_W + [I]_C.
\end{equation}
We have the following formula for the spelling length:
\begin{prop}\label{prop:pairing_in_reduced}
Let $m = \min_{U \in \Vcal(\What)} \Lambda(U)$.  Choose ${\bf U} \in
\Vcalb({\Whatb})$ with optimal reduced pairing $I$ such that
$\lambda(\Ub) = m$ (such a $\Ub$ and $I$ exist by
Proposition~\ref{prop:reduced_exist}). Then every index in
$\Ccal(\Ub)$ is paired in $I$, and
\begin{equation}
  \label{eq:formula}
  \lambda(\Ub) = \sum_{r = 1}^q L({\bf U_r}) - [ I]_W.
\end{equation}
\end{prop}
\begin{proof}
Suppose $c_1 \in \Ccal(\Ub)$ is not paired.  Then either $\bar c_1$
is unpaired, or else it is paired to some $c_2 \in \Ccal(\Ub)$; note
that, since $I$ is reduced, $c_2$ cannot belong to $\Wcal(\Ub)$.
Continue in this way to generate a sequence $c_1, \ldots, c_d$, where
$c_1$ and $\bar c_d$ are unpaired while $\bar c_{r-1}$ is paired to
$c_r$ for  $1 < r \le d$. Remove the set of indices $R = \{c_1, \bar
c_1, \ldots, c_d, \bar c_d\}$ from $\Ub$ and $I$ to obtain $\Vb \in
\Vcalb(\Whatb)$ with valid  pairing $J$.  Then  $L(\Vb) = L({\Ub}) -
2d$ and $| J | = | I | - (d-1)$.
From (\ref{eq:pairing_and_lambda}), it follows that
\begin{equation}\label{eq:U and U''}
  \lambda(\Vb) \le L(\Vb) - 2| J | =
L({\bf U}) - 2| I | - 2 =
  \lambda({\bf U}) - 2 < m,
\end{equation}
in contradiction to the fact that $\lambda(\Vb) \ge m$ for all $\Vb
\in \Vcalb(\Whatb)$.
It follows that all indices in $\Ccal(\Ub)$ are paired.  Therefore
\begin{equation}
  \label{eq:pr_eq_eq}
  2|I | = [I ]_W +  [I ]_C =  [I ]_W+  |\Ccal({\Ub})| =
[I]_W +  L(\Ub) - \sum_{r = 1}^q L({\bf U_r}).
\end{equation}
Since $I$ is optimal,
\begin{equation}
  \label{eq:formula_2}
   \lambda(\Ub) = L(\Ub) - 2|I|  =
  \sum_{r = 1}^q L({\bf U_r}) - [I]_W.
\end{equation}
\end{proof}

\subsection{Proof of Propositions~\ref{prop: spelling length} and \ref{prop: spelling length 2}}\label{sec:ABC}

As the proofs of Propositions~\ref{prop: spelling length} and
\ref{prop: spelling length 2} are  similar, we give  details only
for Proposition~\ref{prop: spelling length}.  We then briefly explain the
differences that arise in the proof of Proposition~\ref{prop:
spelling length 2}.

\begin{proof}[Proof of Proposition~\ref{prop: spelling length}]
Let $\bf \Pcal_{n,p}$ denote the set of words
\begin{align}
  \label{eq:CPb_ijk}
  {\bf \Pcal_{n,p}} =
\{
\big(&{\bf h_0}, A^i, B^j, C^k , {\bf h_0^{-1}},\nonumber\\
&{\bf h_1}, A^{-1}, B^{-1}, C^{-1}, {\bf h_1^{-1}}, \ldots,
{\bf h_{n}}, A^{-1}, B^{-1}, C^{-1}, {\bf h_{n}^{-1}},\nonumber\\
&{\bf h_{n+1}},C,B,A, {\bf h_{n+1}^{-1}}, \ldots {\bf h_{n+p}},
C,B,A , {\bf h_{n+p}^{-1}}\big),\ \ {\bf h_t} \in \Lcal_3\}.
\end{align}
Then $\Ub \in \bf \Pcal_{n,p}$ implies that $U \in \Pcal_{n,p}$.
From Proposition~\ref{prop:lambda=Lambda} it follows that
\begin{equation}
  \label{eq:lambda_form_2}
  \min_{U \in \Pcal_{n,p}} \Lambda(U) =  \min_{{\bf U} \in  \bf \Pcal_{n,p}}
  \lambda({\bf U}).
\end{equation}
For $\Ub \in \bf \Pcal_{n,p}$, let
\begin{equation}
  \label{eq:D(Ub)}
  D_{n,p}(\Ub) = \lambda(\Ub) - (i + j + k - (n+p)).
\end{equation}
Then Proposition~\ref{prop: spelling length} is equivalent to
showing that
\begin{equation}
  \label{eq:D>=0}
  D_{n,p}(\Ub) \ge 0.
\end{equation}
%
%

We  refer to the three-tuples $(A^{-1}, B^{-1}, C^{-1})$ in
(\ref{eq:CPb_ijk}) as {\it
  negative triples}, and the three-tuples $(C,B,A)$ as {\it positive triples}.
Let us partition $\Wcal(\Ub)$ into three sets as follows: Let
$\Wcal_0(\Ub)$ denote the set of indices of the
letters of $A^i B^j C^k$, $\Wcal_-(\Ub)$ the set of indices of the
negative triples, and $\Wcal_+(\Ub)$
 the set of indices of the positive triples.  Then
\begin{equation}
  \label{eq:W_decomp}
   \Wcal(\Ub) =\Wcal_0(\Ub) \cup \Wcal_-(\Ub) \cup \Wcal_+(\Ub).
\end{equation}

We proceed by assuming there exists $\Ub \in \bf\Pcal_{n,p}$ with
$D_{n,p}(\Ub) < 0$ and then deriving a contradiction.
  Lemma
\ref{lem:untrue_1} below, whose proof we give at the end of this section,
shows that if there are unpaired indices in
$\Wcal_-(\Ub)$, then we can construct $\Vb \in {\bf \Pcal_{n-1,p}}$ with
$D_{n-1,p}(\Vb) < 0$; similarly, if there are unpaired indices in
$\Wcal_+(\Ub)$, we can construct  $\Vb \in {\bf \Pcal_{n,p-1}}$ with
$D_{n,p-1}(\Vb) < 0$.
\begin{lem}\label{lem:untrue_1}
  Suppose there exists $\Ub \in \bf \Pcal_{n,p}$ with
  $D_{n,p}(\Ub) < 0$.  Let $I$ be an optimal reduced pairing
  for $\Ub$ (which we may assume to exist by virtue of Proposition~\ref{prop:reduced_exist}).
  If $\Wcal_-(\Ub)$ contains indices which are unpaired in $I$, then
  there exists $\Vb \in \bf{\Pcal_{n-1,p}}$ such that
  $D_{n-1,p}(\Vb) < 0$. Similarly,
if $\Wcal_+(\Ub)$ contains indices which are unpaired in $I$, then
  there exists $\Vb \in \bf{\Pcal_{n,p-1}}$ such that
  $D_{n,p-1}(\Vb) < 0$.
\end{lem}
We apply Lemma~\ref{lem:untrue_1} repeatedly, reducing $n$ or $p$ by
one each time, as many times as possible.  We obtain
thereby a word $\Vbp \in \bf \Pcal_{m,q}$, where $0 \le m \le n$ and  $0 \le q
\le p$, such that $D_{m,q}(\Vbp) < 0$.  We may assume that
\begin{equation}\label{eq: m>q}
m > 2q,
\end{equation}
since (\ref{eq: crude lower bound}) implies that $D_{m,q}(\Vbp) \ge
4q-2m$.
By Proposition~\ref{prop:reduced_exist}, there exists a
word $\Vb \in \bf \Pcal_{m,q}$ with optimal reduced pairing $J$ such that $\lambda(\Vb) = \min_{\Ubp \in \Pcalb_{m,q}}
\lambda(\Ubp)$.   Then $\lambda(\Vb) \le \lambda(\Vbp)$, so that  $D_{m,q}(\Vb) < 0$.
We may assume that every index in
$\Wcal_{\pm}(\Vb)$ is paired in $J$, as otherwise we could apply
Lemma~\ref{lem:untrue_1} again.  By Proposition~\ref{prop:pairing_in_reduced},
every index in $\Ccal(\Vb)$ is paired.
Therefore, the only unpaired indices in $\Ical(\Vb)$
belong to $\Wcal_0(\Vb)$.

Let us
remove the set of  unpaired indices from $\Vb$ and $J$ to obtain $\Wb$ and
$K$.  By construction,
every index of $\Wb$ is then paired in $K$.  But if a word has a
pairing in which every index is paired, it follows from
(\ref{eq:pairing_and_lambda}) that $\lambda(\Wb ) = 0$, so that, by
Lemma~\ref{lem:zerolen},
\begin{equation}
  \label{eq:V=e}
  W = [\Wb] = e.
\end{equation}
On the other hand, since $J$ is reduced, every paired index in
$\Wcal_-(\Vb)$ is linked to a distinct index in $\Wcal_+(\Vb) \cup
\Wcal_0(\Vb)$ (recall that linked indices necessarily belong to
inverse letters).  This implies that $\Wcal_0(\Vb)$ has
$3(m - q)$ paired indices, with $(m-q)$ indices corresponding to $A$, $B$ and $C$ each.  Since,
in passing to $\Wb$, no indices in
$\Ccal(\Vb)$ were removed, it follows that $W$ is of the form
\begin{equation}
  \label{eq:Vb}
W= f_0 A^{m-q} B^{m-q} C^{m -q} f_0^{-1}\prod_{j=1}^m
 f_j (CBA)^{-1}  f_j^{-1}
\prod_{k=1}^q
 g_{k} CBA
g_{k}^{-1} = e,
\end{equation}
where $f_j, g_k \in F(A,B,C)$, and we have incorporated (\ref{eq:V=e}).

Let $\Phi$ be the homomorphism from $F(A,B,C)$ to $F(A,B)$ given by
\begin{equation}
  \label{eq:Phii}
  \Phi(A) = A, \quad \Phi(B) = B, \quad \Phi(C) = A^{-1} B^{-1},
\end{equation}
Then  $\Phi(CBA) = \Phi((CBA)^{-1})= e$.  Applying $\Phi$ to
(\ref{eq:Vb}) and conjugating by $\Phi(f_0)^{-1}$, we obtain
\begin{equation}
  \label{eq:V2}
  A^{m-q} B^{m-q} (A^{-1} B^{-1})^{m-q}  = e.
\end{equation}
This implies $m = q$, in contradiction to (\ref{eq: m>q}).
\end{proof}


\begin{proof}[Proof of Proposition~\ref{prop: spelling length 2}]
Let $\bf \Qcal_{n,p}$ denote the set of  words
\begin{align}
  \label{eq:Q_cal}
  {\bf \Qcal_{n,p}} =
\{
\big(&{\bf h_0}, A^i, B^j, C^k , {\bf h_0^{-1}},\nonumber\\
&{\bf h_1}, C^{-1}, B^{-1}, A^{-1}, {\bf h_1^{-1}}, \ldots,
{\bf h_{n}}, C^{-1}, B^{-1}, A^{-1}, {\bf h_{n}^{-1}},\nonumber\\
&{\bf h_{n+1}},A,B,C, {\bf h_{n+1}^{-1}}, \ldots {\bf h_{n+p}},
A,B,C, {\bf h_{n+p}^{-1}}\big),\ \ {\bf h_t} \in \Lcal_3\}.
\end{align}
Then $\Ub \in \bf \Qcal_{n,p}$ implies that $U \in \Qcal_{n,p}$. For
$\Ub \in \bf \Qcal_{n,p}$, let
\begin{equation}
  \label{eq:E(Ub)}
  E_{n,p}(\Ub) = \lambda(\Ub) - (i + j + k - (n+p) - 2).
\end{equation}
Then Proposition~\ref{prop: spelling length 2} is equivalent to
showing that
\begin{equation}
  \label{eq:E>=0}
  E_{n,p}(\Ub) \ge 0.
\end{equation}
Following the proof of Proposition~\ref{prop: spelling length}, we
construct $X \in \Qcal_{m,q}$ with $m > 2q + 1$ such that
\begin{equation}
  \label{eq:X}
X= f_0 A^{m-q} B^{m-q} C^{m -q} f_0^{-1}\prod_{j=1}^m
 f_j (ABC)^{-1}  f_j^{-1}
\prod_{k=1}^q
 g_{k} ABC
g_{k}^{-1} = e.
\end{equation}
Apply the homomorphism $\Psi: F(A,B,C) \rightarrow F(A,B)$ defined
by $\Psi(A) = A$, $\Psi(B) = B$, $\Psi(ABC) = e$ to get that
\begin{equation}
  \label{eq:X2}
 A^{m-q} B^{m-q} ( B^{-1}A^{-1})^{m-q}  = e.
\end{equation}
This implies either $m = q$ or $m = q + 1$,  both of which
contradict $m > 2q + 1$ (as $q \ge 0$).
\end{proof}
%
%
%
%
%
%
%

It remains to give the proof of Lemma~\ref{lem:untrue_1}.
\begin{proof}[Proof of Lemma~\ref{lem:untrue_1}]
From Proposition~\ref{prop:pairing_in_reduced}, we have that
\begin{equation}
  \label{eq:formulab}
D_{n,p}(\Ub) = 4(n+p) - [I]_W.
\end{equation}
Suppose that $\Wcal_-(\Ub)$ contains indices which are unpaired in
$I$ (the argument for the case where $\Wcal_+(\Ub)$ has unpaired
indices is similar). Then at least one of the negative triples in
$\Ub$ contains letters with unpaired indices. Let $t, t+1, t+2$
denote the (consecutive) indices of one such negative triple.  Let
$v$ denote the number of these indices that are paired in $I$, so
that, by assumption, $0 \le v \le 2$. Denote these paired indices by
$t + \alpha_u$, where $1 \le u \le v$ and $\alpha_u = 0, 1\
\text{or} \ 2$.  Let $j_u$ denote the index to which $t + \alpha_u$
is linked. Since $I$ is reduced, we have that
\begin{equation}
  \label{eq:where_ju_is}
  j_u \in \Wcal(\Ub).
\end{equation}
Let $c_{u,1}, \ldots, c_{u,m_u}$ denote the linking sequence from
$j_u$ to $t+ \alpha_u$, so that
\begin{equation}
  \label{eq:which_are_paired}
  \{j_u, c_{u,1}\}, \{\bar c_{u,1}, c_{u,2}\}, \ldots,  \{\bar
  c_{u,m_{u-1}}, c_{u,m_u}\},  \{\bar c_{u,m_u}, t + \alpha_u\} \in
  I
\end{equation}
(of course, if $v=0$, there are no such  linking sequences).
Let
\begin{equation}
  \label{eq:R_def_1}
  R = \{t, t+1, t+2,\} \cup \left\{\cup_{u=1}^v \{c_{u,1}, \bar c_{u,1}, \ldots, c_{u,m_u}, \bar
c_{u,m_u}\}\right\}.
\end{equation}
%
Let us remove $R$ from $\Ub$ and $I$ to get $\Vb$ with pairing $J$.
Since the indices in $R\cap \Ccal(\Ub)$ occur in conjugate pairs and
$\Vb$ has one less negative triple than does $\Ub$, it follows that
$\Vb \in \bf \Pcal_{n-1,p}$.

We argue that every index in $\Ccal(\Vb)$ is paired in $J$, ie
\begin{equation}
  \label{eq:JC_paired}
  [J]_C = |\Ccal(\Vb)|
\end{equation}
(note that $J$ need not be optimal and we have not established that
$J$ is reduced, so this assertion does not follow from
Proposition~\ref{prop:pairing_in_reduced}). Take $d \in \Ccal(\Vb)$,
and let $c = \psi_R^{-1}(d)$.  Then
\begin{equation}
  \label{eq:where_c_is}
  c \in  \Ccal(\Ub) - R.
\end{equation}
Since $I$ is optimal and reduced, it follows from
Proposition~\ref{prop:pairing_in_reduced} that $c$ is paired to some
index $b$ in $I$.  We claim that $b \notin R$.  Given that this is
so, it follows from (\ref{eq:J_construct}) that $d$ is indeed paired
in $J$, to $\psi_R(b)$ in fact.

To show that $b \notin R$, let us assume the
contrary; 
then we have that $b \in R$ is paired to $c \notin R$.  Examination
of (\ref{eq:R_def_1}) and (\ref{eq:which_are_paired}) shows that the
only indices in $R$ that are paired to indices not in $R$ are the
$c_{u,1}$'s, which are paired to the $j_u$'s.  It follows that $b =
c_{u,1}$ and $c = j_u$ for some $1 \le u \le v$.  But this would
imply that $j_u = c \in \Ccal(\Ub)$,
in contradiction to (\ref{eq:where_ju_is}).

Since $J$ need not be an optimal pairing for $\Vb$, we have the
inequality (rather than equality)
\begin{multline}
  \label{eq:lambda(V)}
  \lambda(\Vb) \le L(\Vb) - 2|J| = |\Wcal(\Vb)| +
  |\Ccal(\Vb)| -  [J]_W -  [J]_C \\
=   |\Wcal(\Vb)| - [J]_W  = i + j + k + 3(n-1 + p) - [J]_W,
\end{multline}
where we have used
  (\ref{eq:JC_paired}) in the second equality. It follows that
\begin{equation}
  \label{eq:D(Vb)}
  D_{n-1,p}(\Vb) \le
  4(n-1 + p) -   [ J ]_W.
\end{equation}

From (\ref{eq:J_construct}), the only paired indices in $\Wcal(\Ub)$
which are not mapped by $\psi_R$ into paired indices in $\Wcal(\Vb)$
are those which belong to $R$ itself and those which are paired with
indices in $R$.  There are $v$ of the former -- namely the $t+
\alpha_u$'s -- and $v$ of the latter -- namely the $j_u$'s, where $1
\le u \le v$.  Therefore,
\begin{equation}
  \label{eq:<J>_W}
  [J]_W = [I]_W - 2v.
\end{equation}
From (\ref{eq:formulab}), (\ref{eq:D(Vb)}) and (\ref{eq:<J>_W}), we
conclude that
\begin{equation}
  \label{eq:D(Vb_b)}
  D_{n-1,p}(\Vb) \le D_{n,p}(\Ub) - 2(2 - v)  < 0,
\end{equation}
as $v \le 2$ and $ D_{n,p}(\Ub) < 0$, by assumption.
\end{proof}

\end{document}